\documentclass[11pt,notitlepage]{article}

\usepackage{booktabs} 

\usepackage{fullpage}
\usepackage{theorem}
\usepackage{graphicx}
\usepackage{subfig}
\usepackage{float}
\usepackage{afterpage}
\usepackage{amsmath,amssymb}
\usepackage{color}
\usepackage{url}
\usepackage{bbm}
\usepackage{epsfig}
\usepackage{epstopdf}

\usepackage{multirow}
\usepackage{array}
\usepackage[mathscr]{eucal}
\usepackage{chngpage}
\usepackage{amsmath}

\newtheorem{theorem}{Theorem}
\newtheorem{lemma}{Lemma}

\newtheorem{remark}{Remark}

\newcommand{\abs}[1]{\left| #1 \right|}

\newcommand{\E}{\mathcal{E}}
\newcommand{\var}{\mathbf{Var}}
\newcommand{\cov}{\mathbf{Cov}}

\renewcommand{\Pr}{\mathbf{Pr}}

\newcommand{\D}{\mathcal{D}}
\renewcommand{\O}{\mathcal{O}}

\renewcommand{\S}{\mathcal{S}}
\renewcommand{\P}{\mathcal{P}}

\renewcommand{\E}{\mathbf{E}}

\newcommand{\genoa}{{\tt GEN50kS}}

\newcommand{\genob}{{\tt GEN20kS}}
\newcommand{\genoc}{{\tt GEN20kM}}
\newcommand{\genod}{{\tt GEN20kL}}
\newcommand{\genoe}{{\tt GEN80kS}}
\newcommand{\genof}{{\tt GEN320kS}}
\newcommand{\trec}{{\tt TREC}}
\newcommand{\uniref}{{\tt UNIREF}}

\newcommand{\ebdjoin}{{\tt EmbedJoin}}

\newcommand{\pass}{{\tt PassJoin}}

\newcommand{\qchunk}{{\tt QChunk}}
\newcommand{\vchunk}{{\tt VChunk}}
\newcommand{\mjoin}{{\tt MinJoin}}

\newcommand{\ED}{\text{ED}}

\renewcommand{\paragraph}[1]{\medskip \noindent {\bf #1.}}

\newenvironment{proof}{{\noindent \bf Proof:}}{\hfill$\square$}

\usepackage{algpseudocode}
\usepackage{algorithmicx}
\usepackage{algorithm}

\algnewcommand\algorithmicforeach{\textbf{for each}}
\algdef{S}[FOR]{ForEach}[1]{\algorithmicforeach\ #1\ \algorithmicdo}

\usepackage{multirow}
\usepackage{array}
\usepackage[mathscr]{eucal}
\usepackage{chngpage}

\usepackage{authblk}
\usepackage{blindtext}

\begin{document}

\title{MinJoin: Efficient Edit Similarity Joins \\via Local Hash Minima\thanks{Authors are supported in part by NSF CCF-1525024, IIS-1633215 and CCF-1844234.}}

\author{Haoyu Zhang\thanks{Email: hz30@umail.iu.edu} }
\author{Qin Zhang\thanks{Email: qzhangcs@indiana.edu} }
\affil{Department of Computer Science, Indiana University}
\date{}

\maketitle

\begin{abstract}
We study the problem of computing similarity joins under edit distance on a set of strings.  Edit similarity joins is a fundamental problem in databases, data mining and bioinformatics. It finds important applications in data cleaning and integration, collaborative filtering, genome sequence assembly, etc.  This problem has attracted significant attention in the past two decades. However, all previous algorithms either cannot scale well to long strings and large similarity thresholds, or suffer from imperfect accuracy. 

In this paper we propose a new algorithm for edit similarity joins using a novel string partition based approach.  We show mathematically that  with high probability our algorithm achieves a perfect accuracy, and runs in linear time plus a data-dependent verification step.  Experiments on real world datasets show that our algorithm significantly outperforms the state-of-the-art algorithms for edit similarity joins, and achieves perfect accuracy on all the datasets that we have tested.
\end{abstract}

\section{Introduction}
\label{sec:intro}

Edit similarity joins is a fundamental problem in the database and data mining literature, and finds numerous applications in data cleaning and integration, collaborative filtering, genome sequence assembly, etc.  In this problem we are given a set of strings $\{s_1, \ldots, s_n\}$ and a distance threshold $K$, and asked to output all pairs of strings $(s_i, s_j)$ such that $\ED(s_i, s_j) \le K$, where $\ED(\cdot, \cdot)$ is the edit distance function, which is defined to be the minimum number of insertions, deletions and substitutions to transfer one string to another.  There is a long line of research on edit similarity joins~\cite{GJKMS01,AGK06,BMS07,BHSH07,LLL08,XWL08,WLF10,QWL11,WQX13,LDW11,WLF12}.  

A major challenge for most existing algorithms, as pointed out by the recent work \cite{ZZ17}, is that they do not scale well to long strings and large edit thresholds.  Long strings and large thresholds are critical for applications involving long sequence data such as big documents and DNA sequences, where a small threshold $K$ may just give zero output. 
For example, in the genome sequence assembly, in which the first step is to find all pairs of similar reads under edit distance, the third generation sequencing technology such as single molecule real time sequencing (SMRT)~\cite{Roberts2013} generates reads of 1,000-100,000 bps long with 12-18\% sequencing errors (i.e., percentage of insertions, deletions and substitutions).  Large threshold is also identified as the main challenge in a recent string similarity search/join competition~\cite{WDG14}, where it was reported that ``an error rate of 20\%-25\% pushes today’s techniques to the limit''.

Different from previous algorithms which are deterministic and return the exact answers, in \cite{ZZ17} the authors proposed a randomized algorithm named \ebdjoin\ which is more efficient on long strings and large thresholds.  However, the accuracy (more precisely, the {\em recall}, 
i.e., the number of pairs found by the algorithm divided by the total number of similar pairs; 
the {\em precision} of all algorithms discussed in this paper is always 100\%) of \ebdjoin\ is only  95\% - 99\% on a number of real-world datasets tested in \cite{ZZ17}. The imperfect accuracy is inherent to \ebdjoin\ which we shall explain shortly.  The main question we are going to address in this paper is: 
\begin{quote}
{\em Can we solve edit similarity joins efficiently on long string and large edit threshold while achieving perfect accuracy with a good probability?}
\end{quote}

\paragraph{Our Contribution}
We propose a novel randomized algorithm named \mjoin\ to address the above question. The high level framework of \mjoin\ is simple: it first partitions each string into a set of substrings, and then uses hash join on these substrings to find all pairs of strings that share at least one common substring.  At the end a verification step is used to remove all false positives. Our string partition scheme works as follows: We first assign each letter $\alpha$ in the string $s$ a value, which is a random hash value of the $q$-gram ($q$ is a value determined by the string length, the threshold $K$, and the size of the alphabet) starting from $\alpha$. We then determine the {\em anchors} of string $s$ using the following strategy: a letter $\alpha$ is an anchor if and only if its value is the smallest among all letters in a certain neighborhood of $\alpha$.  At the end we simply partition $s$ at all of its anchors. 

Via a rigorous mathematical analysis we can show that under our partition scheme, with a good probability, any pair of strings with edit distance at most $K$ will share at least one common partition. We can also show that this partition procedure runs in {\em linear} time.

We have verified the effectiveness of \mjoin\ by an extensive set of experiments.  Though in our experiments we do not include a parallel repetition step which is for the purpose of guaranteeing that our algorithm achieves perfect accuracy with high probability in theory (see the discussion in Section~\ref{sec:partition-analysis}), our experimental results show that \mjoin\ is able to achieve perfect accuracy on all datasets that were used in \cite{ZZ17}.  Moreover, \mjoin\ is faster than all existing exact (deterministic) algorithms by orders of magnitudes on datasets of long strings and large edit thresholds, and is also faster than \ebdjoin\ by a good margin.

\paragraph{Previous Work and Comparisons}
Many of the existing algorithms on edit similarity joins also follow the string partition framework. The performance of the algorithm is largely determined by the number of partitions generated for each string, and the number of queries made to the indices (e.g., hash tables) to search for similar strings. 

We discuss several state-of-the-art algorithms according to the experimental studies in \cite{JLFL14}.

\qchunk~\cite{QWL11} is an exact edit similarity join algorithm based on string partition. \qchunk\ first obtains a global order $\sigma$ of $q$-grams. It then partitions each string into a set of chunks with starting positions $1, q+1, 2q+1, \ldots$, and stores the first $K+1$ chunks (according to the order $\sigma$) in a hash table.  Next, for each string the algorithm queries the hash table with the string's first $N - (\lceil( N - K )/q \rceil - K) + 1$ $q$-grams according to $\sigma$ to check if there is any match, where $N$ is the string length.~\footnote{Alternatively, for each string we can store the first $N - (\lceil( N - K )/q \rceil - K) + 1$ $q$-grams in the hash table, and make queries with the first $K+1$ chunks.} 

\pass~\cite{LDW11} is another exact algorithm based on string partition.  The algorithm partitions each string $s$ into $K+1$ equal-length segments, and records the $i$-th segment into an inverted index $L_{\abs{s}}^i$.  Next, for each string the algorithm queries some of the inverted indices to find similar strings; the number of queries made for each string is $\Theta(K^3)$, which is $\Theta(N^3)$ when $K$ is a fixed percentage of $N$.

\vchunk~\cite{WQX13} is the one that is closest to \mjoin\ among all algorithms that we are aware of. In \vchunk\ each string is partitioned into at least $2K+1$ chunks of possibly different lengths, 

determined by a {\em chunk boundary dictionary (CBD)}. More precisely, each string is cut at positions of appearances of each word in CBD to obtain its chunks.  The CBD is data dependent and the optimal one is NP-hard to compute.  In \cite{WQX13} the authors proposed a greedy algorithm for computing a CBD in time $O(n^2 N^2 / K)$, where $n$ is the number of input strings, and $N$ is the maximum string length. 

The recently proposed algorithm \ebdjoin~\cite{ZZ17} uses a very different approach.  \ebdjoin\ first embeds each string from the edit distance metric space to the Hamming distance metric space, translating the original problem to finding all pairs of strings that are close under Hamming distance. It then uses Locality Sensitive Hashing to compute (approximate) similarity joins in the Hamming space.  However, the embedding algorithm employed by \ebdjoin\ has a worst case distance distortion $K$, which can be very large. Although in practice the distortion is much smaller, it still contributes a non-negligible percentage of false negatives which prevent a perfect accuracy.

Compared with these existing algorithms, \mjoin\ has the following major advantages.
\begin{itemize}
\item For each string \mjoin\ only generates $O(K)$ partitions, and makes the same amount of queries (for searching similar strings), which are significantly smaller than \qchunk\ and \pass.  

\item \mjoin\ can compute partitions of all strings in time $O(n N)$, i.e., linear in the input size, which is even faster than the computation of CBD in \vchunk.  

\item \mjoin\ is able to reach perfect accuracy on tested datasets, compared with 95\%-99\% of \ebdjoin.

\end{itemize}

\paragraph{A Comparison with MinHash Based Approach}  We would like to note that \mjoin\ is quite different from the folklore algorithm using MinHash, in which for each string we collect all its $q$-grams and hash them to numbers, and then pick the one with the smallest hash value as the signature for the subsequent hash join; to increase the accuracy we can pick multiple signatures using different hash functions for each string.  

To see the difference, in \mjoin\ the hash values of the $q$-grams are used to partition a string to substrings/signatures, while in the MinHash based approach the $q$-grams are the signatures themselves.  In \mjoin\ we set $q$ to be a small number (more precisely, $q = \Theta(\log_{\abs{\Sigma}} (N/K))$ where $\Sigma$ is the alphabet of the string) 
in order to make all $q$-grams distinct in every small neighborhood of the string. And one partition will give us all the signatures of the string. While in the MinHash based approach, it is not clear how to find the best combination of the value $q$ and the number of signatures (or, hash functions) to use, for the purpose of achieving a perfect accuracy under a small running time.  We are not aware of any theory for guiding the choices of $q$ and the number of signatures in the MinHash based approach for edit similarity joins. In Section~\ref{sec:exp-minhash} we will show experimentally that \mjoin\  significantly performs the MinHash based approach in both accuracy and running time.

\paragraph{More Related Work}
There is a large body of work on similarity joins under edit distance.  A large number of the existing algorithms fall into the category called the {\em signature-based} approach, in which we compute for each string a set of signatures, and then apply various filtering methods to those signatures to select a set of candidate pairs for verification.  All the string partition based algorithms that we have discussed can be thought as special cases of the signature-based approach.  Other algorithms in this category include {\tt GramCount} \cite{GJKMS01}, {\tt AllPair} \cite{BMS07}, {\tt FastSS} \cite{BHSH07}, 
{\tt ListMerger} \cite{LLL08}, {\tt EDJoin} \cite{XWL08}, and {\tt AdaptJoin} \cite{WLF12}.  

There are a few algorithms that use different approaches, including the embedding-based algorithm \ebdjoin\ discussed previously, the  tree-based algorithm {\tt M-Tree}~\cite{CPZ97}, the enumeration-based algorithm {\tt PartEnum}~\cite{AGK06}, and the trie-based algorithm {\tt TrieJoin}~\cite{WLF10}.  However, except \ebdjoin, others' performance is not as good as the best partition-based approaches.

\paragraph{Notations}  We have listed a set of notations to be used in this paper in Table~\ref{tab:notation}.

\begin{table}[t]
\centering
\scalebox{1}{
\begin{tabular}{|p{.1\textwidth}| p{.5\textwidth}| m{.1\textwidth}|} 
\hline
Notation & Definition\\ 
\hline
$[n]$ & $[n] = \{1, 2, \ldots, n\}$ \\
\hline
$K$ & edit distance threshold\\ 
\hline
$\S$ & set of input strings \\ 
\hline
$s_i$ & $i$-th string in $\S$ \\ 
\hline
$n$ & number of input strings, i.e.,  $n = \abs{\S}$\\ 
\hline
$\abs{s}$ & length of string $s$ \\ 
\hline
$s_{i..j}$ & substring of $s$ starting from the $i$-th \\ 
& letter to the $j$-th letter \\ 
\hline
$N$ & maximum string length \\ 
\hline
$\Sigma$ & alphabet of strings in $\S$ \\ 
\hline
$q$ & length of $q$-gram\\ 
\hline
$\Pi$ & random hash function $\Sigma^q \rightarrow(0,1) $\\ 
\hline
$T$ & number of targeted partitions; $T = \Theta(K)$\\ 
\hline
$r$ & radius for computing local minimum \\ 
\hline
\end{tabular}
}
\caption{Summary of Notations}
\label{tab:notation}
\end{table}

\section{A String Partition Scheme Using Local Hash Minima}
\label{sec:oblivious}
In this section we present the string partition algorithm and analyze its properties.

\subsection{The Algorithm}
\label{sec:partition-algo}

\begin{algorithm}[t]
\caption{Partition-String ($s, T, \Pi$)}
\label{alg:partition}
\begin{algorithmic}[1]
\Require Input string $s$, number of targeted partitions $T$,  random hash function $\Pi: \Sigma^q \rightarrow (0,1) $
\Ensure Partitions of $s$: $\P  = \{(pos, len)\}$, where $(pos, len)$ refers a substring of $s$ starting at the $pos$-th position with length $len$

\State  $\P \leftarrow \emptyset $
\State $A  = \{a_1, \dots, a_p\} \leftarrow$ Find-Anchor($s, T, \Pi$) \label{ln:partbasic-1}
\ForEach {$i \in [1, p - 1]$}
	\State $\P \leftarrow  \P \cup (a_{p}, a_{p + 1} - a_{p})$
\EndFor 
\end{algorithmic}
\end{algorithm}

\begin{algorithm}[t]
\caption{Find-Anchor($s, T, \Pi$)}
\label{alg:oblivious}
\begin{algorithmic}[1]
\Require Input string $s$, number of targeted substrings $T$,  random hash function $\Pi: \Sigma^q \rightarrow (0,1) $
\Ensure  The set of anchors $A$ on $s$

\State  $A \leftarrow \{1\} $
\State $r \leftarrow \lfloor \frac{|s| - q + 1 - T}{2T + 2} \rfloor$  \label{ln:r}
\State Initialize an empty array $h$ with $\abs{s} - q + 1$ elements
\ForEach {$i \in [|s| - q + 1]$}
	\State $h[i] \leftarrow \Pi(s_{i .. i + q - 1})$
\EndFor 
\ForEach {$i \in [1 + r, \abs{s} - q + 1 -  r]$} \label{ln:check-1}
\State $Label \leftarrow 1$  
	\ForEach {$j \in [i -  r, i + r] \text{ and } j \neq i$}  \label{ln:inner-1}
		\If{$h[i] \ge h[j]$}
			\State $Label \leftarrow 0$
			\State Exit the for loop  
		\EndIf
	\EndFor   \label{ln:inner-2}
	\If{$Label = 1$}
		\State $A \leftarrow  A \cup \{i\}$
	\EndIf
\EndFor  \label{ln:check-2}
\State $A \leftarrow  A \cup \{\abs{s}\}$ 
\end{algorithmic}
\end{algorithm}

We start by giving some high level ideas of our partition scheme.  As mentioned, in \mjoin\ we first partition each string to a set of substrings, and then find pairs of strings that share at least one common partition as candidates for verification.  Consider a pair of strings $x$ and $y\ (\abs{x} = \abs{y} = N)$ with edit distance $k$. Let $\rho : [N] \to [N] \cup \{\perp\}$ be the {\em optimal} alignment between $x$ and $y$, where $\rho(i) = j \in [N]$ means that either $x[i] = y[j]$ or $x[i]$ is substituted by $y[j]$ in the optimal transformation, and $\rho(i) = \perp$ means that $x[i]$ is deleted in the optimal transformation.  If we pick any $k$ indices $1 < i_1 < \cdots < i_k < N$ such that $\rho(i_\ell) \neq \perp\ (\ell \in [k])$, partition $x$ at indices $i_1, \ldots, i_k$ to $k+1$ substrings, and partition $y$ at indices $\rho(i_1), \ldots, \rho(i_k)$ to $k+1$ substrings, then by the pigeonhole principle $x$ and $y$ must share at least one common partition.  

Of course obtaining an optimal alignment between $x$ and $y$ before the partition is unrealistic.  Our goal is to partition each string independently, while still guarantee that with a good probability, any pair of similar strings will share at least one common partition.  

We present our partition algorithm in Algorithm~\ref{alg:partition} and Algorithm~\ref{alg:oblivious}.  Let us briefly describe them in words.  Algorithm~\ref{alg:partition} first calls Algorithm~\ref{alg:oblivious} to obtain all {\em anchors} (to be defined shortly) of the input string $s$, and then cuts $s$ at each anchor into a set of substrings.  To compute all anchors,  Algorithm~\ref{alg:oblivious} first hashes all the substrings of $s$ of length $q$ (i.e., $s[1..q], s[2..q+1], \ldots$) into values in $(0,1)$.  Now we have effectively transferred $s$ to an array $h[]$ of size $\abs{s} - q + 1$, with each coordinate taking a value in $(0,1)$.  We call a coordinate $i$ in $h[]$ a {\em local minimum} if its value is strictly smaller than all other coordinates within a distance $r$ of $i$ (for a pre-specified parameter $r$, call it the {\em neighborhood size}).  Algorithm~\ref{alg:oblivious} outputs the corresponding $i$-th letter in string $s$ as an anchor.
For convenience, in the rest of the paper we also call a local minimum coordinate in $h[]$ an anchor.  

We will show that for a pair of strings $x, y$, if they share a common substring $\sigma$ that is long enough, then there must be at least two letters $u, v$ in $\sigma$ such that $u$ and $v$ are two adjacent anchors in both $x$ and $y$, which means that if we use anchors to partition $x$ and $y$, then they must share at least one common partition.  On the other hand, we know that for two strings of length $N$ and edit distance at most $K$, they must share at least one common substring of length $(N-K)/(K+1)$.  Thus by properly choosing the neighborhood size $r$ (as a function of the string length and the number of targeted substrings $T$), we can guarantee that two similar strings will share at least one common partition.

\begin{table}[t]
\centering
\begin{tabular}{ |c|c|c|c|c|c| }
  \hline 
  $3$-gram & Value & $3$-gram & Value & $3$-gram &Value \\
  \hline
    CTA & $0.01$ & ACG & $0.39$  & GAA & $0.69$   \\
    \hline
 GCT &  $0.05$ & AAA  &$0.42$ & AAT & $0.74$   \\
    \hline
 TGC & $0.12$  & AAC  & $0.46$ & ATC & $0.77$   \\
    \hline
 TAA &  $0.21$  & CCT & $0.53$  & GTC & $0.83$   \\
 \hline
  ACC & $0.25$ & TCG & $0.58$  & TGG & $0.89$   \\
    \hline
  CGT & $0.31$  & ATC & $0.62$  & GGA & $0.91$   \\
 \hline
 GTG &  $0.33$  & CGA & $0.64$  & GCG & $0.97$   \\
 \hline
\end{tabular}
\caption{Hash values of $3$-grams}  
\label{tab:hashrank}
\end{table}

\paragraph{\bf A Running Example.}
Before analyzing Algorithm~\ref{alg:partition} we first give a running example.  Table~\ref{tab:hashrank} presents the hash values of all $3$-grams in $\S$ under the hash function $\Pi$.  Table~\ref{tab:input} presents a collection of input strings $\S = \{s_1$, $s_2$, $s_3$, $s_4$, $s_5 \}$ and their lengths. We want to find all pairs of strings with edit distance less than or equal to $K = 4$. Table~\ref{tab:partition} presents the partitions of strings obtained by Algorithm~\ref{alg:oblivious} under parameter $T = 3$. We also calculate the neighborhood size $r$ for each string based on its string length and the parameter $T$.

Considering string $s_1$ as an example, its $6$-th $3$-gram ``CTA'' has a smaller hash value than all its neighbors within distance $r = 2$ (i.e., ``TGC'', ``GCT'', ``TAA'', ``AAC''). Thus ``CTA'' is selected as an anchor of $s_1$. Same to the $14$-th $3$-gram ``CTA''. We then partition $s_1$ to $\{$ACGTG, CTAACGTG, CTAACGTA$\}$. We next find that the strings $s_1, s_2$ share a common partition ``CTAACGTG'',  $s_3, s_4$ share a common partition ``TCGAAT'', and $s_3, s_4, s_5$ share a common partition ``CGTCGAAT'', which give the following candidate pairs: $(s_1, s_2)$, $(s_3, s_4)$, $(s_3, s_5)$, $(s_4, s_5)$. After computing the exact edit distance of each pair, we output $(s_1, s_2)$, $(s_3, s_4)$, $(s_3, s_5)$ as the final answer (i.e., those whose edit distances are no more than $K = 4$).

\begin{table}[t]
\centering
\begin{tabular}{ |c|c|c| }
  \hline 
  ID & String & Length \\
  \hline
  $s_1$ & ACGTGCTAACGTGCTAACGTG & $21$ \\
  $s_2$ & AAACGTGCTAACGTGCTAACCT & $22$\\
  $s_3$ & TCGAATCGTCGAATCGTCGAA & $21$\\
  $s_4$ & TCGAATCGTCGAATCGTGGAA  & $21$\\
  $s_5$ & GTGCGAATCGTCGAATCGTCG & $21$\\
 \hline
\end{tabular}

\caption{Input strings}  
\label{tab:input}
\end{table}

\begin{table}[t]
\centering
\begin{tabular}{ |c|c|c| }
  \hline 
  ID & Partitions of string & $r$ \\
  \hline
  $s_1$ &  ACGTG, CTAACGTG, CTAACGTA  & $2$ \\
  $s_2$ & AAACGTG, CTAACGTG, CTAACCT  & $2$\\
  $s_3$ &  TCGAAT, CGTCGAAT, CGTCGAA & $2$\\
  $s_4$ &  TCGAAT, CGTCGAAT, CGTGGAA  & $2$\\
  $s_5$ &  GTGCGAAT, CGTCGAAT, CGTCG  & $2$\\
 \hline
\end{tabular}
\caption{Partitions of strings by Algorithm~\ref{alg:partition} ($T = 3$)}  
\label{tab:partition}
\end{table}

\paragraph{Discussions}
We would like to discuss two items in more detail.  First, we require the value of an anchor in the hash array $h[]$ to be {\em strictly} smaller than its $2r$ neighbors. The purpose of this is to reduce the number of false positives generated by periodic substrings with short periods; false positives will increase the running time of the verification step of the \mjoin\ algorithm.  In real world datasets, periodic substrings are often caused by  systematic errors, and may be shared among different strings.  For example, consider the following periodic substring on genome data ``\ldots AAAAAAAA \ldots'' produced by sequencing errors, if we allow the  value of an anchor to be equal to its neighbors, then we may have many anchors in this substring.  Consequently, two strings both containing such a substring will be considered as a candidate pair even that they are very different elsewhere.

Second, we use different neighborhood size $r$ for strings of different lengths. More precisely, we set $r = \lfloor \frac{\abs{s}-q+1-T}{2T + 2} \rfloor$ where $T = \Theta(K)$ is an input parameter standing for the number of targeted partitions. The purpose of doing this, instead of choosing a fixed $r$ for all strings, is again to reduce false positives. Indeed, if we choose the same $r$ for all strings, then long strings will generate many partitions, since in order to achieve perfect accuracy we cannot set $r$ to be too large at the presence of short strings. Consequently, the large number of partitions generated by long strings will contribute to many false positives.  

This is in contrast to \vchunk, who cuts the string whenever it finds a word in CBD appearing on the string. Consequently two strings of very different length but sharing a relatively long substring are likely to be considered as a candidate pair, producing a false positive for the verification.

\subsection{The Analysis}
\label{sec:partition-analysis}

We now analyze the properties of Algorithm~\ref{alg:partition}.  Our goal is to understand how many partitions Algorithm~\ref{alg:partition} will generate (which will contribute to the running time of \mjoin\ as we shall see in Section~\ref{sec:mjoin}), and what is the probability for two similar strings to share a common partition.  

To keep the analysis clean, we assume that in any $r$-neighborhood of the array $h[]$ all the coordinates are distinct, which is true if (1) we assume that all corresponding $q$-grams are different, and (2) the hash function $\Pi : \Sigma^q \to (0,1)$ does not produce a collision when applying to $q$-grams.  The later can be easily satisfied if we keep an $O(\log N)$-bit precision ($N$ is the maximum string length) in the range of $\Pi$, in which case there is no hash collision with probability $1 - 1/N^{\Omega(1)}$.  For the former, we set $q = 3 \log_{\abs{\Sigma}} (N/T)$.  Note that by our choice of $r$ we have $r \approx N/(2T)$.  If all letters in a substring of size $r$ are random, then the probability that two $q$-grams in this substring are the same is $1/\abs{\Sigma}^q = \left(\frac{T}{N}\right)^3$.  By a union bound with probability $1 - o(1)$ all $q$-grams in a substring of size $2r$ are different.  We emphasize that this assumption is only used for the convenience of the analysis, and Algorithm~\ref{alg:partition} works without this constraint.

The following lemma states that the number of anchors produced by Algorithm~\ref{alg:oblivious} is concentrated around $T$, the number of targeted partitions.  

\begin{lemma}
\label{lem:oblivious-partition}
Given an input string and a parameter $T$, for any $c > 0$, the number of anchors generated by Algorithm~\ref{alg:oblivious}, denoted by $X$, satisfies
$ \Pr[\abs{X - T} \ge \sqrt{cT}] < 1/c.
$
\end{lemma}

\begin{proof}
Consider the array $h[1..\abs{s}-q+1]$ constructed in Algorithm~\ref{alg:oblivious}; $h[i]$ is the hash value of the $i$-th $q$-gram of $s$.  Let $w = \abs{s}-q+1 - 2r$. For $i = 1, \ldots, w$, define a random variable $X_i$ whose value is $1$ if $h[i+r]$ is the smallest coordinate in the {\em window} $h[i .. i+2r]$, and $0$ otherwise.   Let $X = \sum_{i \in [w]} X_i$, which is the total number of anchors generated by Algorithm~\ref{alg:oblivious}.  We now analyze the random variable $X$.  

We start by computing its expectation. Recall that we have set $r$ to be $\lfloor \frac{|s| - q + 1 - T}{2T + 2} \rfloor$ at Line~\ref{ln:r} of Algorithm~\ref{alg:oblivious}.  For simplicity we ignore the floor operation whose effect is negligible to the analysis.
\begin{eqnarray} 
\E[X] &=&   \sum_{i\in [w]} \E[X_i] =  \sum_{i\in [w]} \Pr[X_i = 1] = \frac{w}{2r + 1} = T. \label{eq:a-0}
\end{eqnarray} 
We next compute the variance. 
\begin{eqnarray} 
\var[X] &=&  \sum_{i \in [w]} \var[X_i] + \sum_{i \neq j} \cov[X_i, X_j] \nonumber \\
 &=& \sum_{i \in [w]} \var[X_i] +  \frac{1}{2} \sum_{i }\sum_{j \neq i } \cov[X_i, X_j]. \label{eq:a-1}
\end{eqnarray} 
We compute the two terms of (\ref{eq:a-1}) separately.  For the first term,
\begin{eqnarray}
\sum_{i \in [w]} \var[X_i] &=& \sum_{i \in [w]} \left(\E[X_i^2] - \left(\E[X_i]\right)^2\right) \nonumber \\
&=& w \times \left( \frac{1}{2r + 1} - \frac{1}{(2r + 1)^2} \right) \nonumber \\
&\le& \frac{w}{2r + 1}. \label{eq:a-2}
\end{eqnarray}
For the second term of (\ref{eq:a-1}), by the definition of the covariance,
\begin{eqnarray*}
\cov[X_i, X_j] &=& \E[X_i X_j] - \E[X_i] \E[X_j] \\
&=& \E[X_i X_j] - \frac{1}{(2r+1)^2}.
\end{eqnarray*}
We analyze $\E[X_i X_j]$ in three cases.

\paragraph{Case I} $\abs{i - j} \ge 2r + 1$.  It is easy to see that in this case $X_i$ and $X_j$ are independent, since their corresponding windows $h[i..i+2r]$ and $h[j..j+2r]$ are disjoint.  We thus have $\E[X_i X_j] = \E[X_i] \E[X_j]$, and consequently $\cov[X_i, X_j] = 0$.
\smallskip

\paragraph{Case II} $\abs{i - j} \le r$.  In this case, $h[i+r]$ is inside the window $h[j .. j+2r]$, and symmetrically $h[j+r]$ is inside the window $h[i.. i+2r]$.  Thus if $X_i = 1$ then we must have $X_j = 0$, and if $X_j = 1$ then we must have $X_i = 0$.  Therefore $\E[X_i X_j] = 0$, and consequently $\cov[X_i, X_j] = - \frac{1}{(2r+1)^2}$.

\paragraph{Case III} $r < \abs{i - j} < 2r+1$. The analysis for this case is a bit more complicated.  Consider two windows $W_i = h[i .. i+2r]$ and $W_j = h[j .. j+2r]$ which overlap. We divide their union into three areas; see Figure~\ref{fig:window} for an illustration.  Area 2 denotes the intersection of the two windows, and Area 1 and Area 3 denote the coordinates that are only in $W_i$ and $W_j$ respectively.  It is easy to see that the number of coordinates in Area 1 and Area 3 are equal; let $\alpha\ (r < \alpha < 2r+1)$ denote this number.

\begin{figure}[t]
\centering
\includegraphics[width=0.55\textwidth]{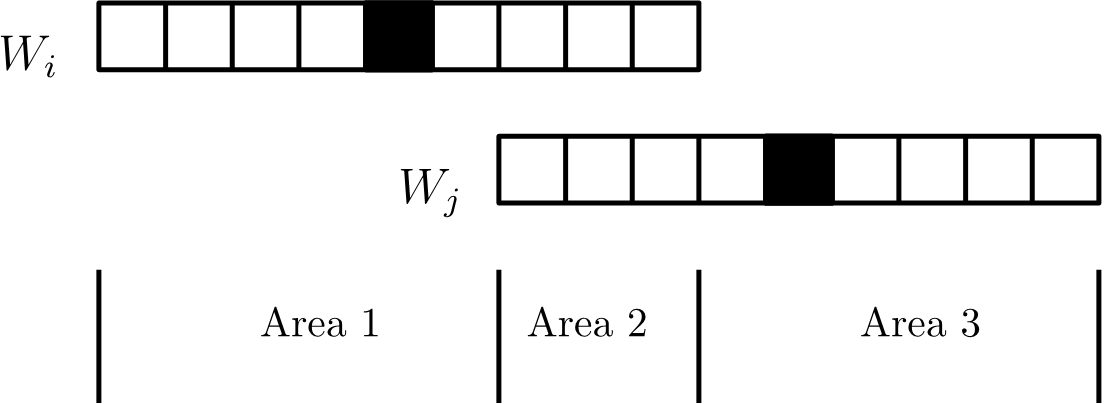}
\caption{Illustration of windows $W_i, W_j$ when $r< \abs{i - j} < 2r + 1$.  Black square represents the central coordinate of the window. The squares in same column correspond to same coordinate in the array $h[]$; we duplicate them for the illustration purpose.}
\label{fig:window}
\end{figure}
	
We write
\begin{eqnarray*}
\E[X_i X_j] &=& \Pr[X_i = 1, X_j = 1] \\
&=& \Pr[X_j = 1\ |\ X_i = 1] \cdot \Pr[X_i = 1] \\
&=& \Pr[X_j = 1\ |\ X_i = 1] \cdot \frac{1}{2r+1}.
\end{eqnarray*}
We thus only need to analyze $\Pr[X_j = 1\ |\ X_i = 1]$.  Define a random variable $Y$ such that $Y = 1$ if the central coordinate of $W_i$ (i.e., $h[i+r]$) is smaller than all coordinates in Area 3.  We have
\begin{eqnarray}
&&\Pr[X_j = 1\ |\ X_i = 1] \nonumber \\
&=& \Pr[X_j = 1\ |\ X_i = 1, Y = 1] \cdot \Pr[Y = 1\ |\ X_i = 1] + \nonumber \\
&& \Pr[X_j = 1\ |\ X_i = 1, Y = 0] \cdot \Pr[Y = 0\ |\ X_i = 1]. \label{eq:b-2}
\end{eqnarray}

Note that $(X_i = 1) \wedge (Y = 1)$ implies that the central coordinate of $W_i$ is smaller than all coordinates in $W_j$, which, however, does not give any information about the relationship between all coordinates in $W_j$.  We thus have
\begin{equation}
\Pr[X_j = 1\ |\ X_i = 1, Y = 1] = \Pr[X_j = 1] = \frac{1}{2r+1}. \label{eq:b-3}
\end{equation} 
On the other hand, $(X_i = 1) \wedge (Y = 0)$ implies that the central coordinate of $W_i$ is smaller than all coordinates in Area $2$, and is larger than some coordinate in Area $3$.  We thus know that the minimum coordinate of $W_j$ must lie in Area $3$.  Therefore $X_j = 1$ if and only if the central coordinate of $W_j$ is larger than all other coordinates in Area $3$.  We get
\begin{equation}
\Pr[X_j = 1\ |\ X_i = 1, Y = 0] = 1/\alpha. \label{eq:b-4}
\end{equation} 
Plugging in (\ref{eq:b-3}) and (\ref{eq:b-4}) to (\ref{eq:b-2}), we have
\begin{eqnarray*}
&&\Pr[X_j = 1\ |\ X_i = 1] \\
&=&  \frac{1}{2r+1} \cdot \Pr[Y = 1\ |\ X_i = 1] +  \frac{1}{\alpha} \cdot \Pr[Y = 0\ |\ X_i = 1] \\
&\le&  \frac{1}{\alpha} \le \frac{1}{r+1}.
 \label{eq:b-5}
\end{eqnarray*}
Consequently we have
\begin{eqnarray*}
\cov[X_i, X_j] \le \frac{1}{2r+1} \cdot \frac{1}{r+1} - \frac{1}{(2r+1)^2} < \frac{1}{(2r+1)^2}.
\end{eqnarray*}
Summing up, we have
\begin{equation} 
    \cov[X_i, X_j]  
    \begin{cases}
      = - \frac{1}{(2r + 1)^2}, & \abs{i - j} \le r \\
      <\frac{1}{(2r + 1)^2}, & r < \abs{i - j} < 2r + 1 \\
      = 0. &  \abs{i - j} \ge 2r + 1
    \end{cases}
    \label{eq:b-6}
\end{equation}
Plugging (\ref{eq:a-2}) and (\ref{eq:b-6}) to (\ref{eq:a-1}), we get
\begin{eqnarray}
\var[X] &<& \frac{w}{2r+1} + \frac{1}{2} \cdot w \cdot 2r \cdot \left(\frac{1}{(2r+1)^2} - \frac{1}{(2r+1)^2} \right)  \nonumber \\
&=& \frac{w}{2r+1} = T.
\label{eq:b-7}
\end{eqnarray}
By (\ref{eq:a-0}), (\ref{eq:b-7}), and the Chebyshev's inequality, we have that for any constant $c > 0$, 
$$
\Pr[\abs{X - T} \ge \sqrt{cT}] < 1/c.
$$
\end{proof}

We have empirically verified the concentration result in Lemma~\ref{lem:oblivious-partition} on two real world datasets (to be introduced in Section~\ref{sec:exp}); see Figure~\ref{fig:partition}. It is clear that the number of partitions Algorithm~\ref{alg:partition} generates are tightly concentrated around the number of target partitions $T$.

\smallskip

\begin{figure}[t]
\centering
\begin{minipage}[d]{0.33\linewidth}
\centering
\includegraphics[width=1\textwidth]{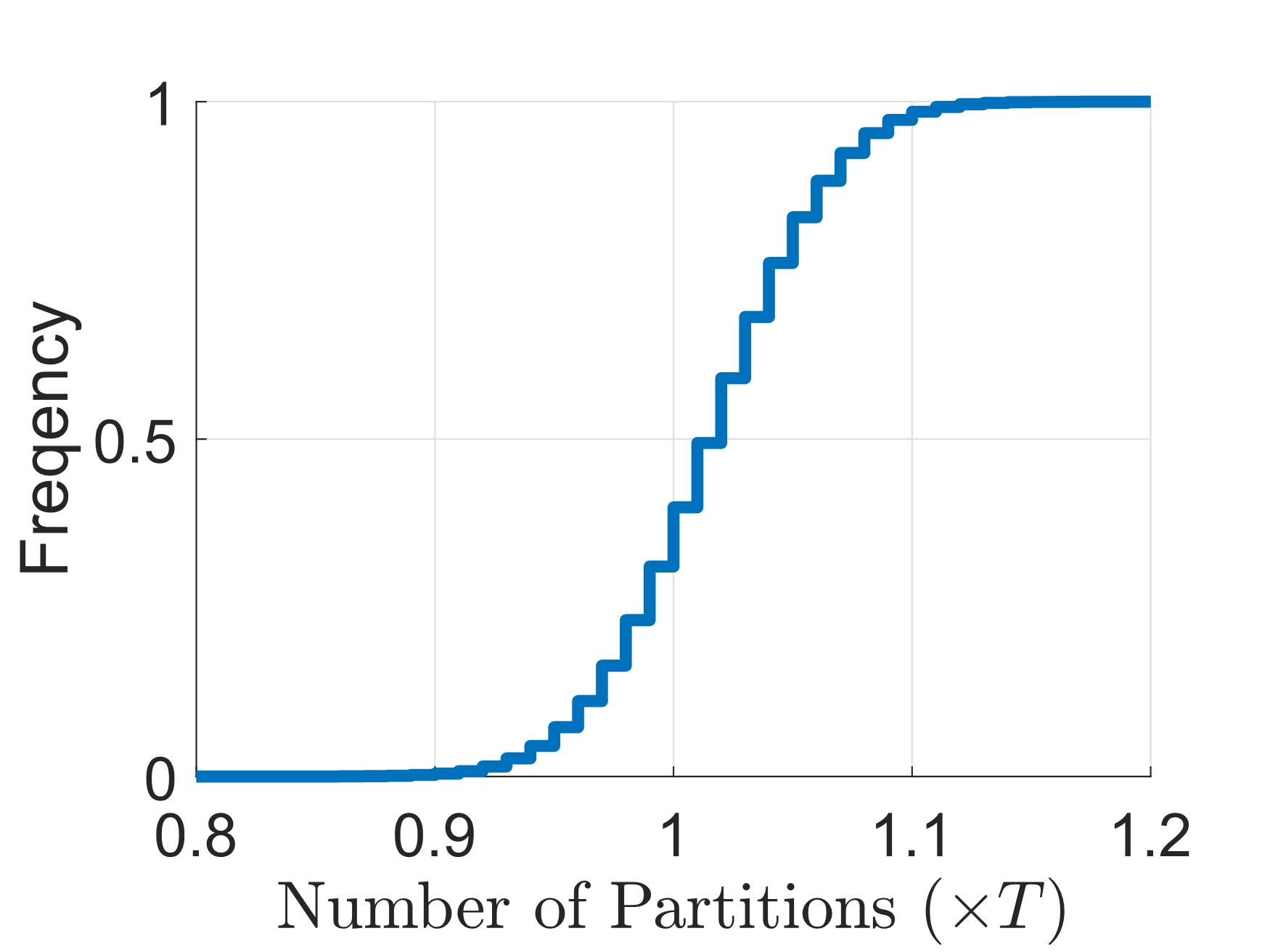}
\centerline{\genoa}
\end{minipage}
\begin{minipage}[d]{0.33\linewidth}
\centering
\includegraphics[width=1\textwidth]{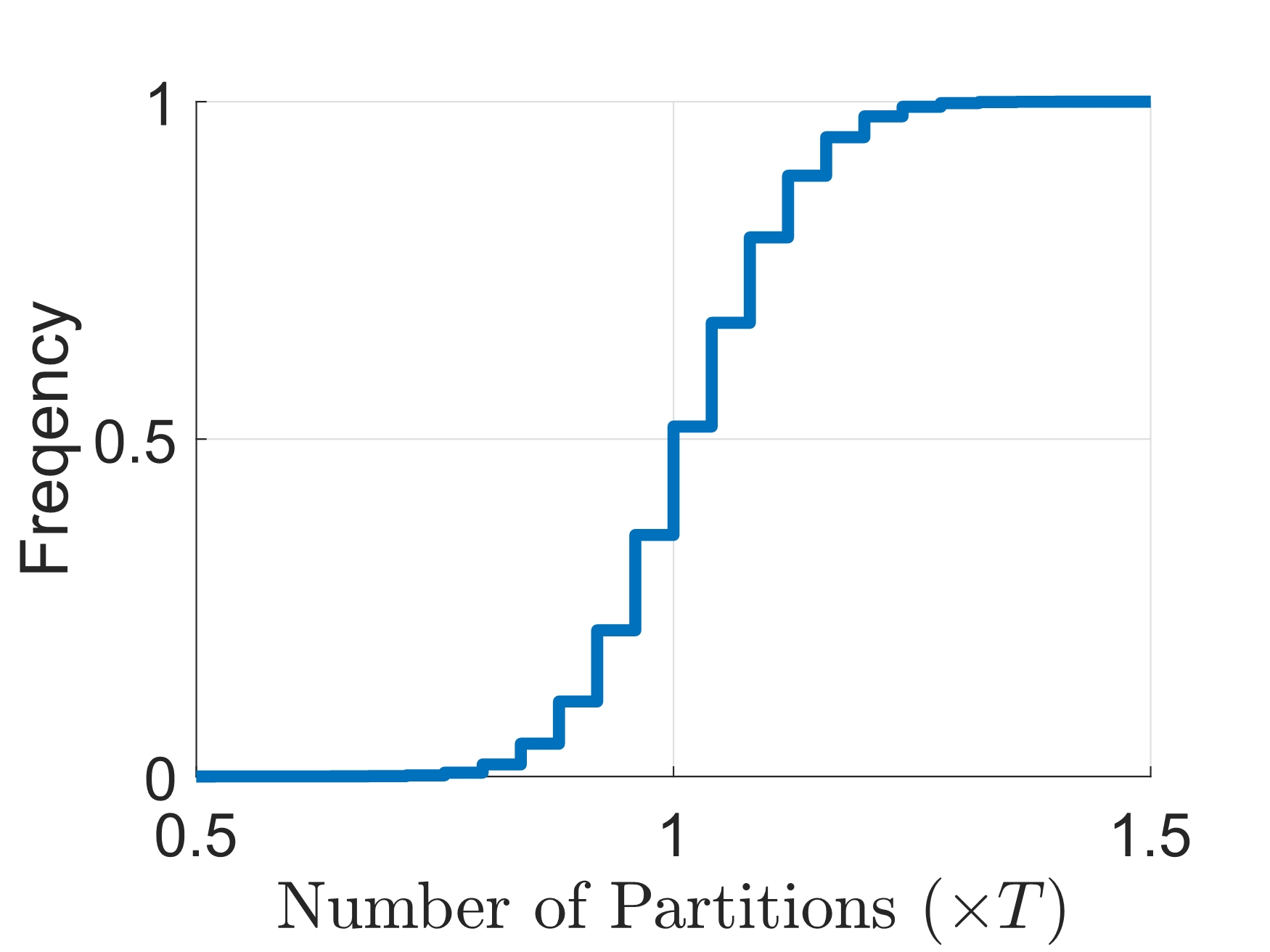}
\centerline{\uniref}
\end{minipage}
\caption{The CDFs of numbers of partitions on each string returned by Algorithm~\ref{alg:partition} on \genoa\ and \uniref\ datasets, with parameters $T = 100$ and $T = 25$ respectively.}
\label{fig:partition}
\end{figure}

We next analyze another key property of our local minimum based partition: Given two similar strings, what is the probability that they share a common partition?  We give the following lemma.  

\begin{lemma}
\label{lem:oblivious-common}
For two strings $s, t$ with $\ED(s, t) \le K$, let $\P_s$ and $\P_t$ be the partitions outputted by Algorithm~\ref{alg:partition} (setting $T = 120K$) on $s$ and $t$ respectively.  Assume $\abs{s} =  \omega(Kq)$.  The probability that $\P_s$ and $\P_t$ share a common partition is at least 0.98.
\end{lemma}

\begin{proof}

Since $\ED(s, t) \le K$, we have $\abs{t} \in [ \abs{s}-K,  \abs{s}+K]$, and $s$ and $t$ must share a common substring of length at least $L = ( \abs{s}-K)/(K+1)$ in the optimal alignment.  

Let $\gamma$ be such a common substring.  Let $r_s = \lfloor \frac{|s| - q + 1 - T}{2T + 2} \rfloor$, and let $\eta = \frac{L-q+1-2r_s}{2r_s+1}$.  When running Algorithm~\ref{alg:oblivious} on $s$, by an almost identical argument as that for the proof of Lemma~\ref{lem:oblivious-partition}, we have that the number of anchors $X$ on $\gamma$ satisfies 
\begin{equation}
\label{eq:c-2}
\Pr[\abs{X - \eta} \ge \sqrt{c \eta}] < 1/c.
\end{equation} 
For $T = 120K$ and $\abs{s} = \omega(Kq)$, we have 
\begin{eqnarray}
\label{eq:c-2-1}
\eta &=& \frac{L-q+1-2r_s}{2r_s+1} \nonumber\\ 
&\ge& \left( \frac{\abs{s}-K}{K+1}-q+1-2r_s \right) \cdot \frac{T+1}{\abs{s}-q+2} \nonumber\\
&\ge& 115.  
\end{eqnarray}   
Plugging (\ref{eq:c-2-1}) to (\ref{eq:c-2}), we have with probability at least $(1 - 1/100) = 0.99$  that
\begin{equation}
X \ge \eta - \sqrt{100\eta} > 4,
\end{equation}
which means that with probability $0.99$ there are at least four anchors on $\gamma$.  

Let $a_1, a_2, a_3, a_4$ be four anchors on $\gamma$ when processing $s$ using Algorithm~\ref{alg:oblivious}.  
Let $r_t = \lfloor \frac{|t| - q + 1 - T}{2T + 2} \rfloor$.  Since $\ED(s,t) \le K$ and $T = 120K$, it holds that $\abs{r_t - r_s} \le 1$.  In the case that $r_t = r_s = r$, $a_2$ and $a_3$ must also be anchors when processing $t$ using Algorithm~\ref{alg:oblivious}, since an anchor is fully determined by a neighborhood of size $r$.

For the case when $\abs{r_t - r_s} = 1$, w.l.o.g., assume that $r_s = r$ and $r_t = r + 1$.  Now the probability that $a_2$ is still an anchor when processing $t$, given the fact that $a_2$ is an anchor when processing $s$, is at least $1 - 1/(r+1)$.  Same argument holds for $a_3$.  Thus with probability $0.99 - 2/(r+1) \ge 0.98$ (note that $r = r_s = \lfloor \frac{|s| - q + 1 - T}{2T + 2} \rfloor = \omega(1)$ given $\abs{s} = \omega(qK)$ and $T = 120K$), $a_2$ and $a_3$ are also anchors when processing $t$.

Finally, observe that once $s$ and $t$ share two adjacent anchors $a_2$ and $a_3$, they must share at least one common partition. 
\end{proof}

\begin{remark}[Choice of $T$]
\label{rem:T}
We note that the choice of $T\ (=120K)$ in Lemma~\ref{lem:oblivious-common} is overly ``pessimistic'' -- it is just for the convenience of analysis. Moreover, we only considered {\em one} pair of common substring of length $L \approx \abs{s}/K$, while the average length of the (at most) $K+1$ pairs of common substrings between $s$ and $t$ in the optimal alignment is at least $\frac{s-K}{K+1} \approx \abs{s}/K$.  A finer analysis which considers all pairs of common substrings in the optimal alignment can reduce the value of $T$ all the way down to a value close to $K$, while still guarantee that $\P_s$ and $\P_t$ share a common partition with a good probability.  However, the analysis is a bit 
cumbersome and we will leave it to the full version of this paper.  The main point of this remark is that in practice we can just set $T \approx K$, or even {\em smaller} since in real-world datasets multiple edits may occur in the same location, which {\em effectively} increases the average length of common substrings.  In our experiments we find that $T \in [K/5, K]$ are good choices for all the datasets we have tested.
\end{remark}

{\em Parallel repetitions for boosting the success probability.\ \ } Though the success probability in Lemma~\ref{lem:oblivious-common} is only $0.98$, and it is only for each pair of similar strings, we can easily boost it to high probability for all pairs of similar strings using parallel repetitions.  We can repeat the partition process for each string for $\log n$ times using independent randomness, and then union all the partitions of the string.  Now for each pair of similar strings, the probability that they share a common partition is at least $1 - 0.02^{\log n} \ge 1 - 1/n^5$. We then use a union bound on the at most $n^2$ pairs of similar strings, and get that the probability that all pairs of similar strings share at least one common partition is at least $1 - 1/n^3$.  We note  in our experiments that we do not need this boosting procedure since a single run of the partition process already achieves perfect accuracy. 

\begin{theorem}
\label{thm:main}
If we apply Algorithm~\ref{alg:partition} augmented by the parallel repetition discussed above on all input strings, then with probability $1 - n^{-\Omega(1)}$, all pair of strings with edit distances at most $K$ will share at least one common partition.  The expected running time of the algorithm is $\log n$ times the input size, and the space needed is also $\log n$ times the input size.
\end{theorem}

\begin{proof}
The correctness follows directly from Lemma~\ref{lem:oblivious-common} and the discussion of parallel repetition above.  In the rest of the proof we focus on the time and space.  In fact, to show the claimed time and space usage we can just show that the time and space for partitioning one string $s$ (by Algorithm~\ref{alg:partition}) is linear in terms of the string length $\abs{s}$.

The running time of Algorithm~\ref{alg:partition} is dominated by that of its subroutine Algorithm~\ref{alg:oblivious}.  
The hash values of all $q$-grams of $s$ can be computed by the Rabin-Karp algorithm (the rolling hash) in $O(\abs{s})$ time.   For Line~\ref{ln:check-1}-\ref{ln:check-2} of Algorithm~\ref{alg:oblivious}, since each number in $h[]$ is a random hash value, the inner for-loop (Line~\ref{ln:inner-1}-\ref{ln:inner-2}) runs in $O(1)$ time in expectation.  Therefore the total running time of Algorithm~\ref{alg:partition} is $O(\abs{s})$ in expectation. 

Clearly, the space usage of Algorithm~\ref{alg:partition} is also $O(\abs{s})$.
\end{proof}

\section{The MinJoin Algorithm}
\label{sec:mjoin}

We now present our main algorithm \mjoin, depicted in Algorithm~\ref{alg:mjoin}.  We briefly explain it in words below.

\begin{algorithm}[t]
\caption{MinJoin ($\mathcal{S}, K, T$)}
\label{alg:mjoin}
\begin{algorithmic}[1]
\Require Set of input strings $\mathcal{S} = \{s_1, \ldots, s_n\}$, distance threshold $K$, number of targeted partitions $T$ 
\Ensure  $\O \gets \{(s_i, s_j)\ |\ s_i, s_j \in \S; i \neq j; \text{ED}(s_i, s_j) \le K\}$ 

\State  $\O \leftarrow \emptyset$, $\mathcal{C} \leftarrow \emptyset $   \Comment $\mathcal{C}:$ collection of candidate pairs \label{ln:min-1}
\State Pick a hash function $f : \Sigma^* \to \mathbb{N}$ and initialize an empty hash table $\D$ \label{ln:hash}
\State Generate a random hash function $\Pi: \Sigma^q \rightarrow (0,1) $
\State Sort strings in $\mathcal{S}$ first by string length increasingly, and second by the alphabetical order \label{ln:min-2}
\ForEach{$s_i \in \mathcal{S} \text{ (in the sorted order)}$} \label{ln:min-3}
\State $\P \leftarrow$ Partition-String($s_i, T, \Pi$) \label{ln:min-par}
	\ForEach{$(pos, len) \in \P$}
			\ForEach{$(j, pos_j, len_j)$ in the $f((s_i)_{pos .. pos + len - 1})$-th bucket of $\D$}  \Comment{$f(\cdot)$ is the hash function picked at Line~\ref{ln:hash}}
				\If{$\abs{\abs{s_i} - \abs{s_j}} \le K$} \label{ln:min-7}
					\If{$\abs{pos - pos_j} + \abs{(\abs{s_i} - pos) - (\abs{s_j} - pos_j)} \le K$} \label{ln:min-9}
						\State $\mathcal{C}  \leftarrow \mathcal{C}  \cup (s_i, s_j)$ 
					\EndIf
				\Else
					\State Remove $(j, pos_j, len_j)$ from $\D$ \label{ln:min-8}
				\EndIf
			\EndFor
			\State Store $(i, pos, len)$ in the $f((s_i)_{pos .. pos + len - 1})$-th bucket of $\D$
	\EndFor
\EndFor
\State Remove duplicate pairs in $\mathcal{C}$   \label{ln:min-4}

\ForEach {$(x, y) \in \mathcal{C}$} \label{ln:min-5}
	\If{$\text{ED}(x, y) \le K$}
	\State $\mathcal{O} \leftarrow  \mathcal{O} \cup (x, y)$
	\EndIf
\EndFor \label{ln:min-6}
\end{algorithmic}
\end{algorithm}

The \mjoin\ algorithm has three stages: initialization (Line~\ref{ln:min-1} - \ref{ln:min-2}), join and filtering (Line~\ref{ln:min-3} - \ref{ln:min-4}) and verification (Line~\ref{ln:min-5} - \ref{ln:min-6}).   In the first stage, we initialize an empty set $\mathcal{C}$ for candidate pairs and an empty hash table $\D$, generate a random hash function $\Pi$, and sort all strings according to their lengths for the pruning.  

In the join and filtering stage, we compute the partitions for each input string using Algorithm~\ref{alg:partition}. For each partition  $(pos, len)$, which refers the substring of $s_i$ with length $len$ and $pos$ is the index of its first character on $s_i$, we find all tuples $(j, pos_j, len_j)$ in $f((s_i)_{pos .. pos + len - 1})$-th bucket of hash table $\D$ (that is, we perform a hash join). We use two rules to prune the candidate pairs we have found.  The first condition (Line~\ref{ln:min-7}) says that if the lengths of $s_i$ and $s_j$ differ by larger than $K$, then it is impossible to have $\ED(s_i, s_j) \le K$. Consequently it is impossible to have $\ED(s_j, s_{i'}) \le K$ for any $i' > i$.  

The second condition (Line~\ref{ln:min-9}) concerns the following scenario: if $s_i$ and $s_j$ match at indices $pos$ and $pos_j$, which divides both strings into two substrings $\nu_1 = (s_i)_{1 .. pos - 1}, \nu_2 = (s_i)_{pos..\abs{s_i}}$, and $\mu_ 1 =(s_j)_{1..pos_j - 1}, \mu_ 2 = (s_j)_{pos_j..\abs{s_j}}$.  If $pos$ and $pos_j$ are indeed matched in the optimal alignment, then we must have 
$
\ED(\nu_1, \mu_1) + \ED(\nu_2, \mu_2) \le K,
$
in which case we have
$\abs{(\abs{s_i} - pos) - (\abs{s_j} - pos_j)} +  \abs{pos - pos_j} \le K.$

We add all pairs of strings that pass the two filtering conditions to the candidate set $\mathcal{C}$, and then perform a deduplication step at the end since each pair can potentially be added into $\mathcal{C}$ multiple times. 

In the verification stage, we verify whether each pair of strings in $\mathcal{C}$ indeed have edit distance at most $K$, using the standard dynamic programming algorithm by Ukkonen~\cite{Ukkonen85}.  Due to this verification step our algorithm will never output any false positive.  On the other hand, by Theorem~\ref{thm:main}, if we augment the string partition scheme with parallel repetition, then \mjoin\ will not produce any false negative with probability $1 - 1/n^{\Omega(1)}$.  Therefore \mjoin\ will achieve perfect accuracy with probability $1 - 1/n^{\Omega(1)}$.

\paragraph{Time and Space Analysis}  Let $N$ be the maximum string length in the set of input strings $\S$, and $n = \abs{\S}$. By Theorem~\ref{thm:main} the running time of the partition (without the parallel repetition) is bounded by $O(n N)$.

The total number of pairs that are fed into the filtering steps (Line~\ref{ln:min-7}, \ref{ln:min-9}) inherently depends on the concrete dataset.  Suppose partitions of all strings are evenly distributed into $\abs{\D}$ buckets of the hash table $\D$ (this is indeed what we have observed in our experiments), then we can upper bound this number by $O\left(\frac{nK}{\abs{\D}}\right)^2$ with probability $0.99$. To see this, by the proof in Lemma~\ref{lem:oblivious-partition} we know that the expected number of partitions of each string is $T = \Theta(K)$. By linearity of expectation, the expected number of partitions of all $n$ strings is $n T$. Therefore the total number of actual partitions is bounded by $O(nK)$ with probability $0.99$ by a Markov inequality.
The verification step can be done in $O(\abs{\mathcal{C}} N K)$ where $\mathcal{C}$ is the set of the candidate pairs.  

The space usage is clearly bounded by $O(n N)$, that is, the size of the input.

\begin{theorem}
\label{thm:mjoin}

The \mjoin\ algorithm has the following theoretical properties.  Consider the case that we augment the string partition procedure at Line~\ref{ln:min-par} with $\log n$ parallel repetitions. 
\begin{itemize}
\item  It achieves 100\% accuracy with probability $1 - 1/n^{\Omega(1)}$. 

\item Assuming that the partitions of all strings are evenly distributed into the buckets of the hash table, the running time of \mjoin\ is bounded by
$$O\left(n N \log n + \left(\frac{nK}{\abs{\D}}\right)^2 +\abs{\mathcal{C}} N K \right)
$$
with probability $0.99$, where $\mathcal{C}$ is the set of the candidate pairs \mjoin\ produces before the verification step.

\item The space usage of \mjoin\ is $\log n$ times the size of input.
\end{itemize}
\end{theorem}

\section{Experiments}
\label{sec:exp}

In this section we present our experimental studies. We start by describing the datasets and algorithms used in our experiments.  We then provide a detailed study of the performance of \mjoin.  Finally, we compare \mjoin\ with the state-of-the-art algorithms for edit similarity joins.

\subsection{Setup of Experiments}
\label{sec:setup}

We implemented our algorithms in C++ and performed experiments
on a Dell PowerEdge T630 server with 2 Intel Xeon E5-2667 v4 3.2GHz CPU with 8 cores each, and 256GB memory.

\paragraph{Datasets}
We use the datasets in \cite{ZZ17} which are publicly available.\footnote{See the documentation from the project website of \cite{ZZ17}: \url{https://github.com/kedayuge/Embedjoin}}  Table~\ref{tab:stat} describes the statistics of tested datasets.  

\begin{description}
\item[\uniref:] A dataset consists of UniRef90 protein sequence data obtained from UniProt Project.\footnote{\url{http://www.uniprot.org/}} The sequences whose lengths are smaller than 200 are removed, and the first 400,000 protein sequences are extracted.

\item[\trec:] A dataset consists of titles and abstracts from 270 medical journals.  The title, author, and abstract fields are extracted and concatenated.  Punctuation marks are converted into white space and all letters are in uppercase.

\item[{\tt GEN-X-Y}'s:] Datasets contain 50 human genomes obtained from the Personal Genomes Project,\footnote{\url{https://www.personalgenomes.org/us}} where {\tt X} denotes the number of strings (range from 20k to 320k), and {\tt Y} denotes the string length (S $\approx$ 5k, M $\approx$ 10k, L $\approx$ 20k).  Each string is a substring randomly sampled from the Chromosome 20 of human genome.
\end{description}

\begin{table}[t]
\centering 
\scalebox{1}{
\begin{tabular}{lccccc} 
\hline
Datasets &$n$ &Avg Len &Min Len & Max Len & $|\Sigma|$\\  \hline 
\uniref &400000  &445 &200  &35213 &25\\  \hline
\trec &233435  &1217 &80   &3947 &37    \\  \hline
\genoa &50000   &5000  &4829 & 5152 &4   \\  \hline
\genob &20000   &5000  &4829 &5109 &4   \\  \hline
\genoc &20000   &10000  &9843 &10154 &4   \\  \hline
\genod &20000   &20000  &19821 &20109 &4   \\  \hline
\genoe &80000   &5000  &4814  &5109 &4   \\  \hline
\genof &320000   &5000  &4811  &5154 &4   \\  \hline
\end{tabular}
}
\caption{Statistics of tested datasets (from \protect \cite{ZZ17})}
\label{tab:stat}
\end{table}

\paragraph{Algorithms}
We compare \mjoin\ with the state-of-the-art algorithms for edit similarity joins discussed in the  introduction, including \pass \cite{LDW11}, \qchunk \cite{QWL11}, \vchunk  \cite{WQX13}, \ebdjoin \cite{ZZ17}. All codes are downloaded from the corresponding project websites.

\paragraph{Measurements and Choices of Parameters}
We use three metrics to measure the performance of tested algorithms: time, space, and accuracy.  

We note that except \mjoin\ and \ebdjoin\ which are randomized and may have false negatives, all other tested algorithms are deterministic and output the exact number of similar pairs, and thus their accuracy is always 100\%.   According to our theoretical analysis (Theorem~\ref{thm:main} and Remark~\ref{rem:T}), by setting $T$ appropriately and using $\log n$ repetitions of the string partition procedure (Algorithm~\ref{alg:partition}), \mjoin\ can output all similar pairs with a high probability.  In practice, we found that a single execution of Algorithm~\ref{alg:partition} with $T \in [K/5, K]$ can already achieve 100\% accuracy on all tested datasets.\footnote{Whenever there is an exact algorithm that finishes in a reasonable amount of time so that we get to know the ground truth.} In fact, as we shall see in Figure~\ref{fig:tacc} and Figure~\ref{fig:ttime}, varying $T$ in this range will not change the accuracy by much, but it does slightly affect the running time since larger $T$ will introduce more false positives for verification.

In the rest of this section we will always write the accuracy for \ebdjoin\ on the plots, and omit that for \mjoin\ if it is 100\%.

We always choose the {\em best} parameters of other tested algorithms. \qchunk\ has two parameters: $q$ (the size of $q$-gram) and indexing method. We found that the {\em indexchunk} always performs better than {\em indexgram} on all datasets, and we always choose the best $q$ for each experiment. \vchunk\ has a parameter {\em scale} to tune. \pass\ has no parameter. \ebdjoin\ has three parameters $m, r, z$. We choose the parameters based on the recommendation of \cite{ZZ17}: We select the best combinations of parameters to achieve at least $95\%$ accuracy on \uniref\ and \trec\ datasets, and at least $99\%$ accuracy on \genoa\ dataset; and we select $r = z= 7, m = 15 - \lfloor \log_2 x \rfloor$ on the rest of datasets, where $x\%$ is the edit threshold. 

Each result is an average of $5$ independent runs.  For \mjoin\ we fix the randomness at the beginning so that all runs return the same result on the same dataset.

\begin{figure}[t]
\centering
\begin{minipage}[d]{0.33\linewidth}
\centering
\includegraphics[width=1\textwidth]{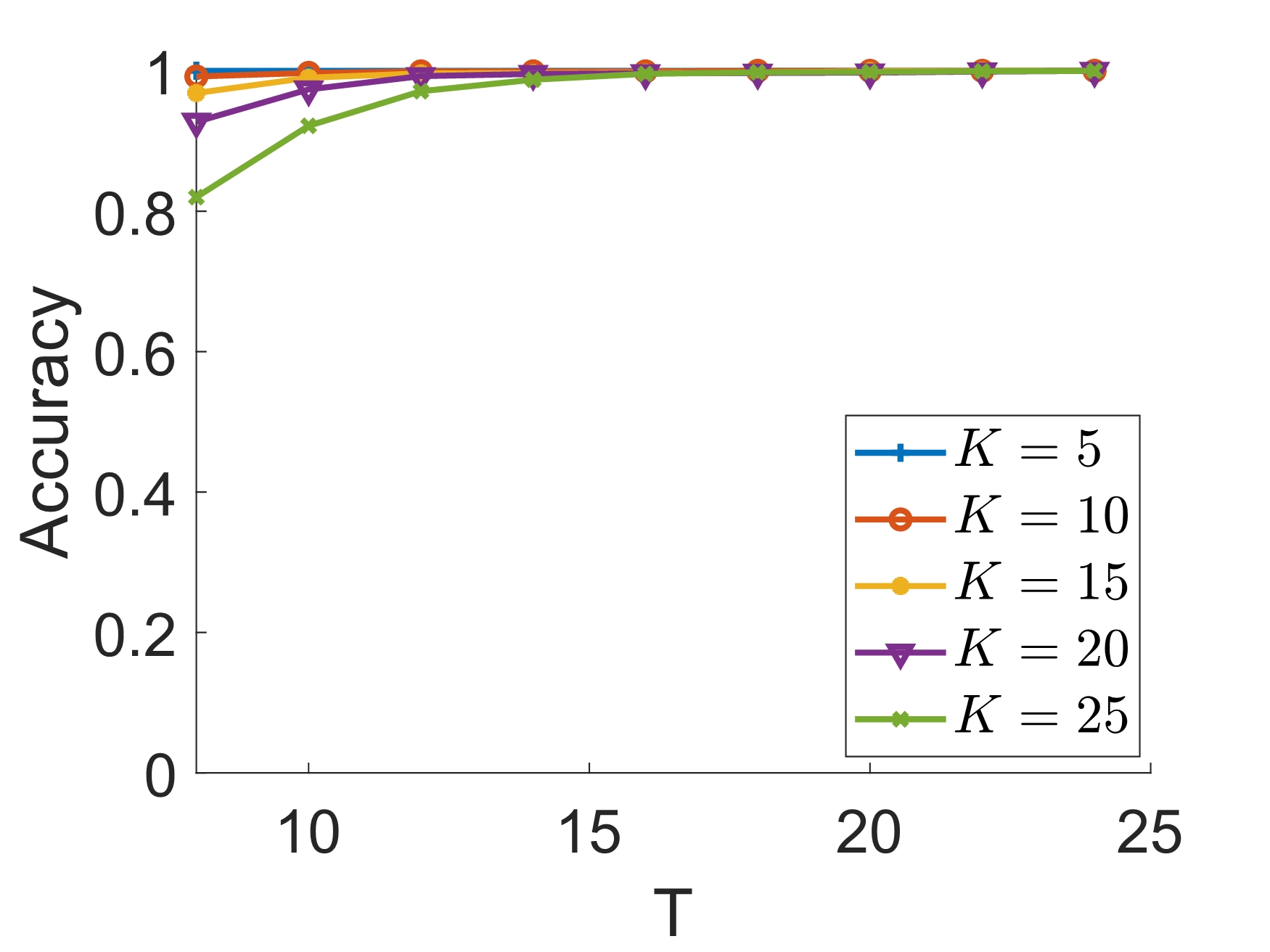}
\centerline{\uniref }
\end{minipage}
\begin{minipage}[d]{0.33\linewidth}
\centering
\includegraphics[width=1\textwidth]{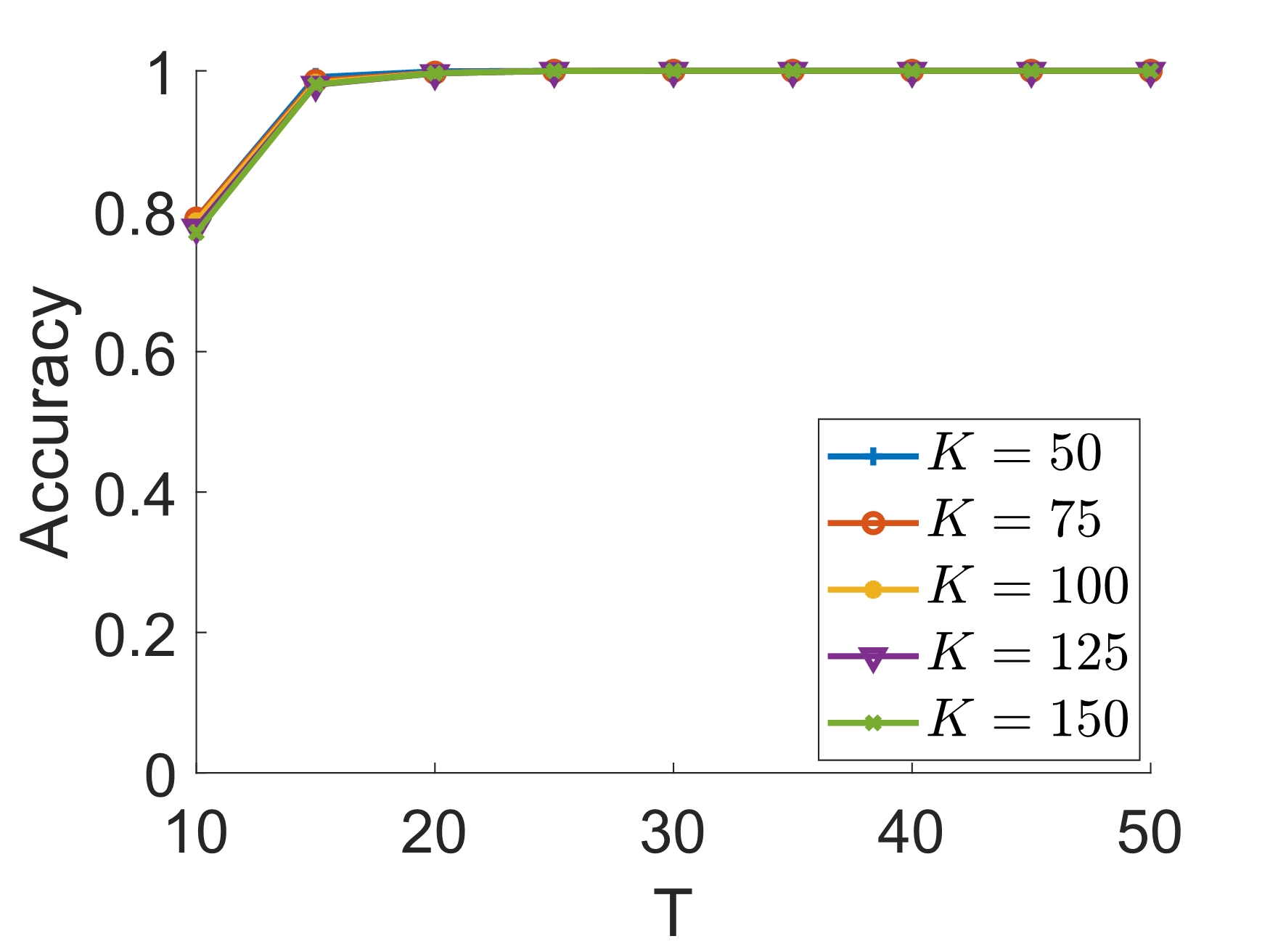}
\centerline{\genoa}
\end{minipage}
\caption{Influence of $T$ on accuracy}
\label{fig:tacc}
\end{figure}

\begin{figure}[t]
\centering
\begin{minipage}[d]{0.33\linewidth}
\centering
\includegraphics[width=1\textwidth]{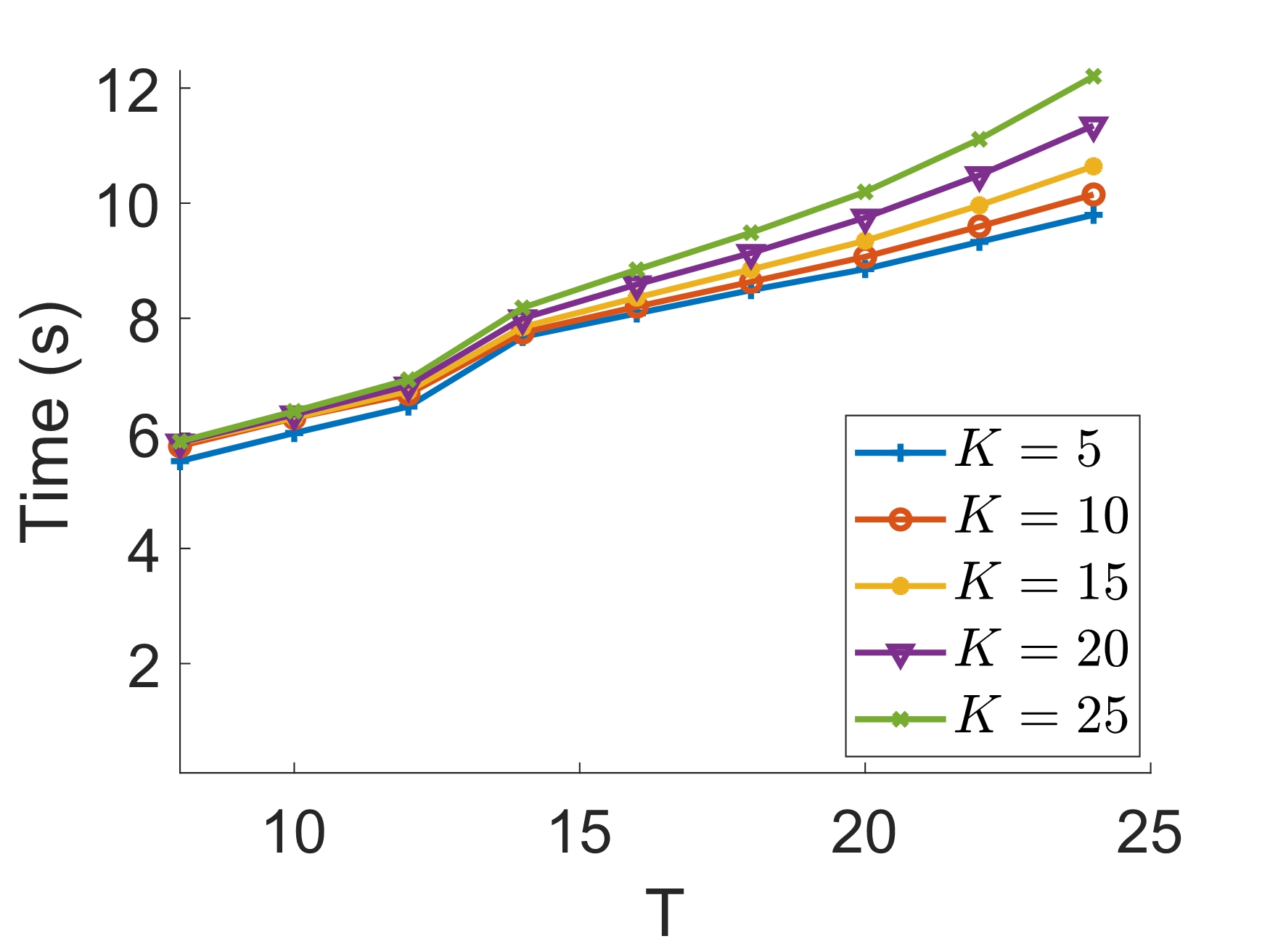}
\centerline{\uniref }
\end{minipage}
\begin{minipage}[d]{0.33\linewidth}
\centering
\includegraphics[width=1\textwidth]{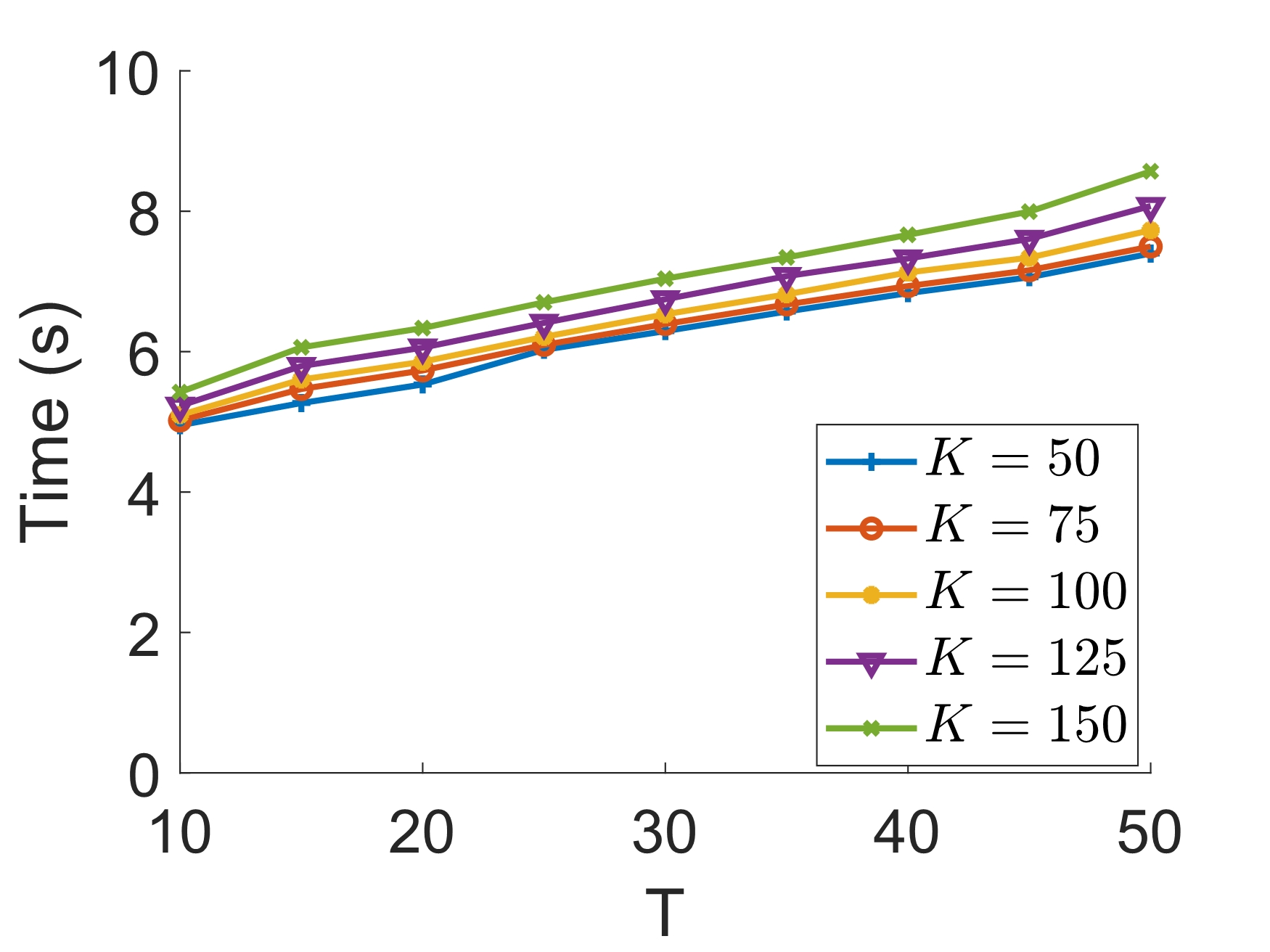}
\centerline{\genoa}
\end{minipage}
\caption{Influence of $T$ on running time}
\label{fig:ttime}
\end{figure}

\begin{figure}[t]
\centering
\begin{minipage}[d]{0.33\linewidth}
\centering
\includegraphics[width=1\textwidth]{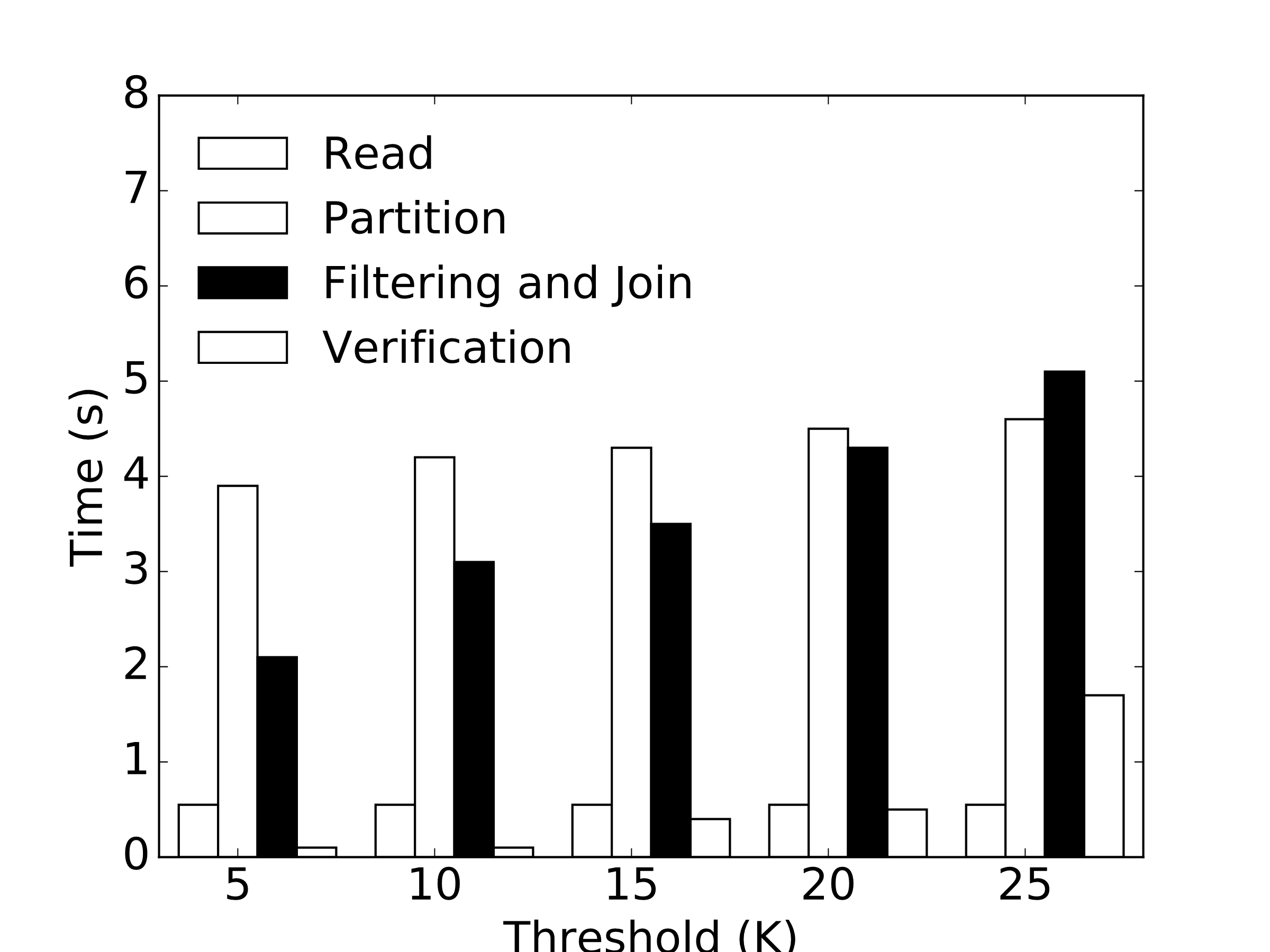}
\centerline{\uniref}
\end{minipage}
\begin{minipage}[d]{0.33\linewidth}
\centering
\includegraphics[width=1\textwidth]{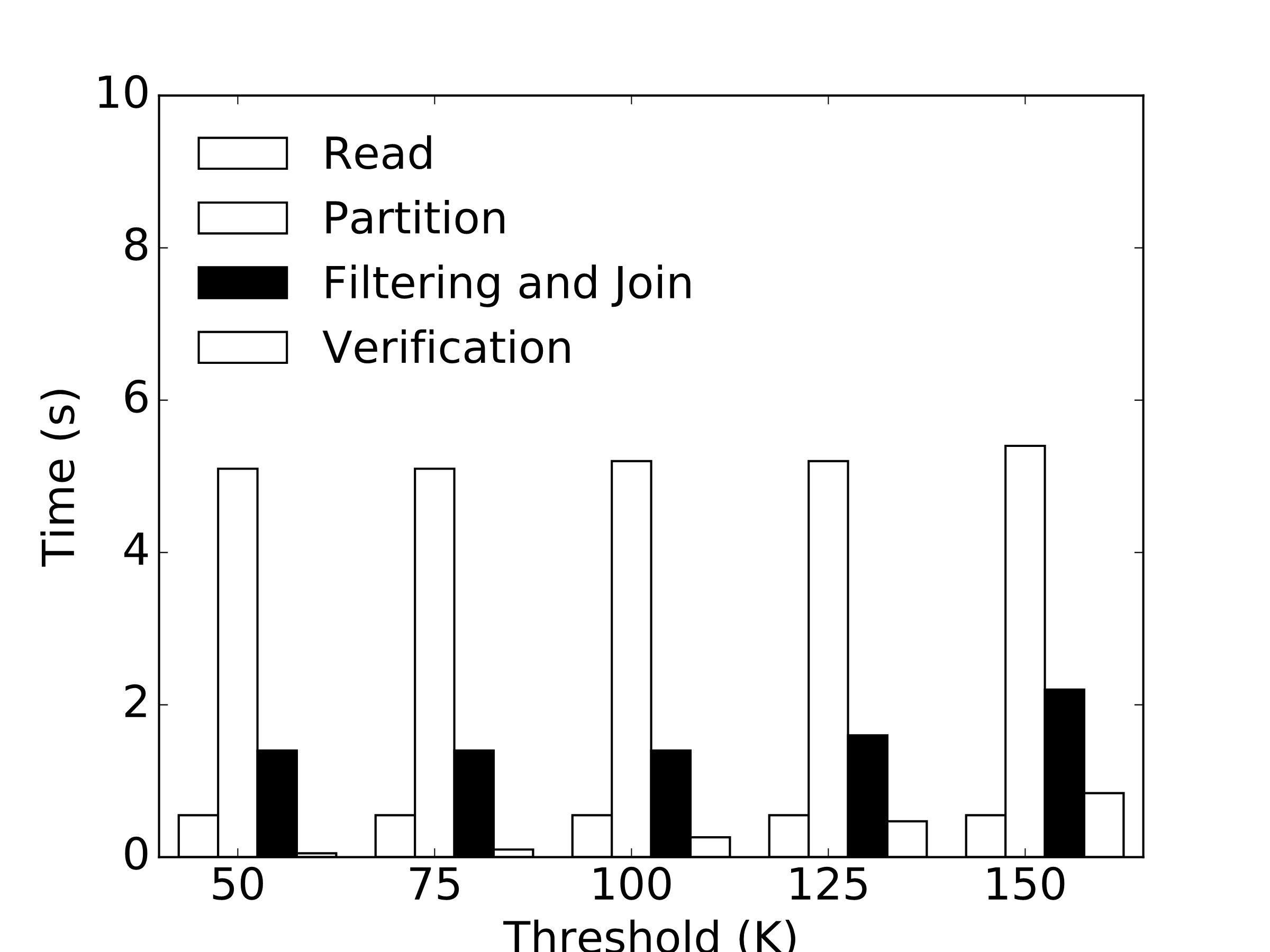}
\centerline{\genoa}
\end{minipage}
\caption{Running time of different parts of \mjoin, varying $K$.}
\label{fig:kbar}
\end{figure}

\subsection{Experiments for \mjoin}
\label{sec:exp-mjoin}

We first show the performance of \mjoin.  We will start by investigating the influence of parameter $T$ on running time and accuracy, and then present the running time of different stages of \mjoin.

\paragraph{Influence of Parameter $T$}  
We study empirically how parameter $T$ influences the accuracy and the running time of \mjoin.  We present the influence of $T$ on the accuracy and running time in Figure~\ref{fig:tacc} and \ref{fig:ttime} respectively.  As predicted by theory, both time and accuracy increase when $T$ increase. We also tested different edit thresholds $K$.  We observe that when $K$ is larger, we need a larger $T$ to maintain the 100\% accuracy, which is also consistent with the theory where we need to pick $T = \Theta(K)$.  As mentioned in Section~\ref{sec:setup}, we found that setting $T$ in the range $[K/5, K]$ is good for all the tested datasets.

\paragraph{Running Time of Different Parts of \mjoin}  
We have also measured the running time of different parts of \mjoin, including input read, string partition, hash join and filtering, and verification. We present in Figure~\ref{fig:kbar} the running time of \mjoin\ on (1) reading the input strings, (2) partitioning strings, (3) performing the hash join and filtering, and (4) verification varying the edit threshold $K$. Certainly, the input read time will not change for different $K$. We observe that the time for join and filtering increases slightly when $K$ increases, that for partition is stable, and that for verification increases considerably when $K$ increases. On \uniref\ dataset, the string partitioning as well as join and filtering steps are bottleneck, and on \genoa\ dataset, the string partition step is bottleneck. The verification step takes the smallest amount of time in most cases. 

\subsection{A Comparison with MinHash}
\label{sec:exp-minhash}

\begin{figure}[t]
\centering
\begin{minipage}[d]{0.33\linewidth}
\centering
\includegraphics[width=1\textwidth]{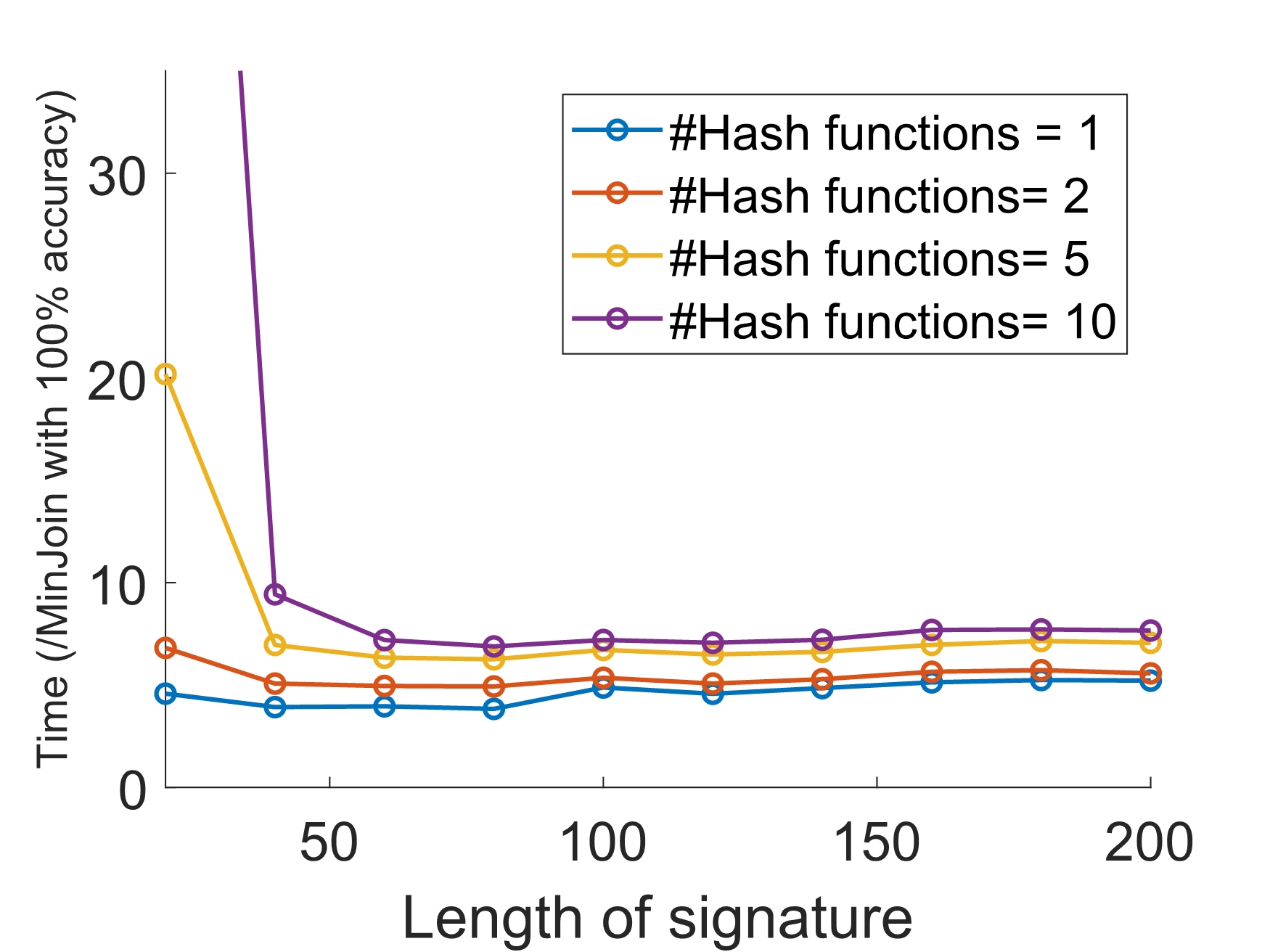}
\centerline{(a) Running time}
\end{minipage}
\begin{minipage}[d]{0.33\linewidth}
\centering
\includegraphics[width=1\textwidth]{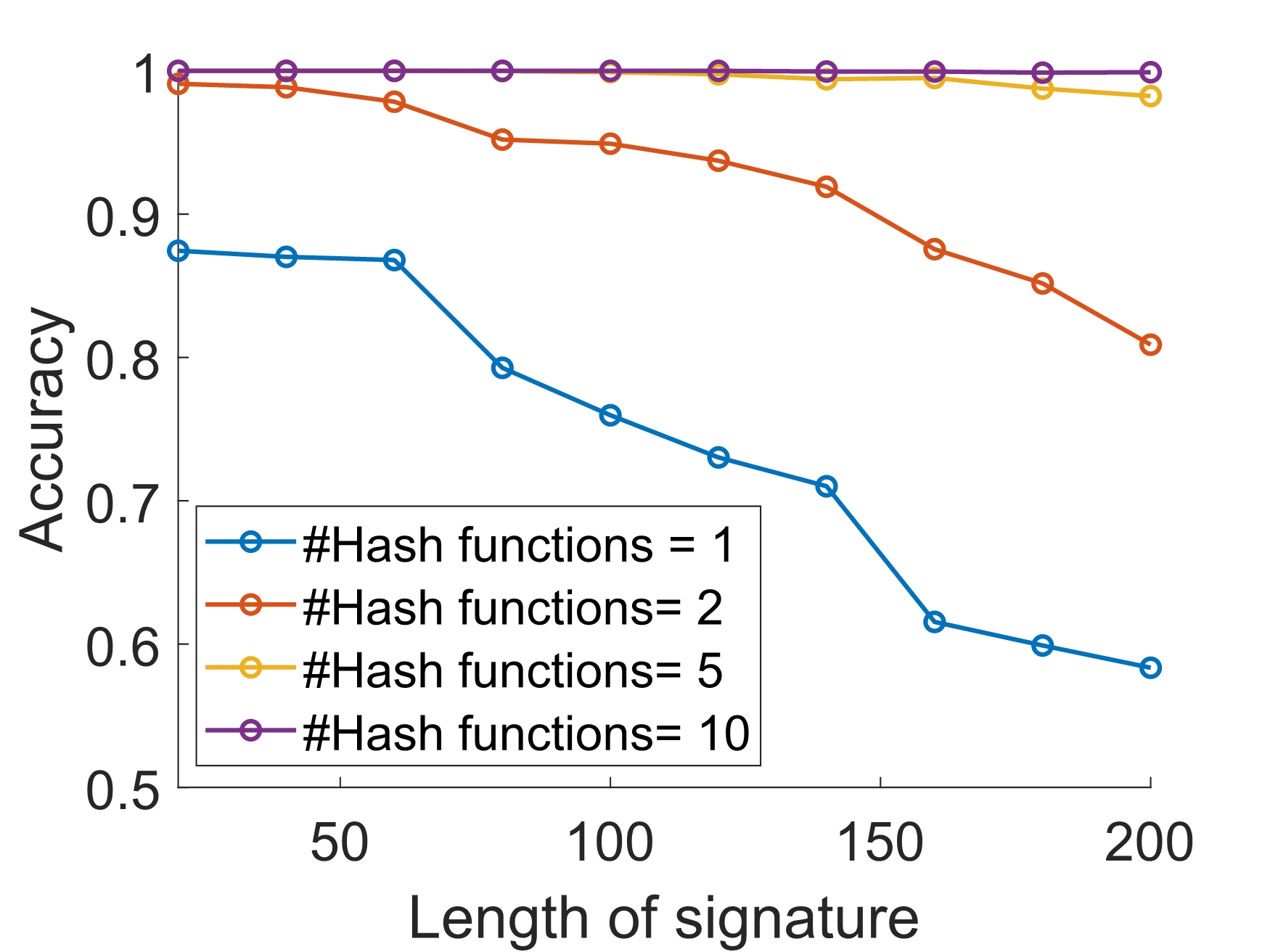}
\centerline{(b) Accuracy}
\end{minipage}
\caption{Performance of the MinHash based algorithm on \genoa\ dataset with $K = 100$. (a) The running time of the MinHash based algorithm as a multiple of that of \mjoin\ at 100$\%$ accuracy. (b) The accuracy of the MinHash based algorithm.}
\label{fig:min}
\end{figure}

Before going to the main body of the experimental study, we try to argue that the folklore MinHash based algorithm is not competitive with \mjoin. The reason that we discuss it separately is that this folklore algorithm has two parameters for which we do not have any guideline for the tuning. We thus try to present its performance by testing different combinations of these parameters.

As mentioned in the introduction, the MinHash based algorithm is straightforward: we convert each string into a set which consists of the hash values of all $q$-grams of the string, and then pick the smallest value as the signature of the string for the subsequent hash join. To boost the accuracy, we can use $\ell$ such MinHash functions, and get $\ell$ signatures for each string.  Applying $\ell$ hash functions to get the signatures is expensive.  A standard optimization method is to use only one hash function, and then select the top-$\ell$ smallest hash values as the signatures.  This is what we use in our experiments.

Figure~\ref{fig:min} shows the running time and accuracy of the MinHash based algorithm when varying the number of hash signatures $\ell$ and the length of signature $q$. The running time is shown as a multiple of \mjoin\ at 100$\%$ accuracy.  We find that the running time and accuracy of the MinHash based algorithm depend on the two parameters $q$ and $\ell$: When increasing parameter $\ell$, both running time and accuracy increase; when increasing parameter $q$, the running time first decreases and then increases a little bit, and the accuracy decreases. We observe the accuracy and running time are sensitive to parameters, and there is no principle on how to select them for edit similarity joins. This is in contrast to \mjoin\ where the only parameter is $T$ (the targeted number of partitions), and we have already discussed how to choose $T$ both theoretically and practically.  Moreover, even we choose the best combination of $\ell$ and $q$, the running time of the MinHash based algorithm is still at least $5$ times of that of \mjoin\ at $100\%$ accuracy. We thus conclude that \mjoin\ outperforms the MinHash based algorithm in all aspects.

\subsection{A Comparison with the State-of-the-Art}
\label{sec:compare}

We now compare \mjoin\ with the state-of-the-art algorithms for edit similarity joins (\qchunk, \pass, \vchunk\ and \ebdjoin).  We will make use of \uniref, \trec\ and \genoa\ for a basic comparison.  These datasets are of modest size so that all algorithms can finish within 24 hours.  We then use larger genome datasets to test the scalability of all algorithms.  

\begin{figure*}[t]
\centering
\begin{minipage}[d]{0.32\linewidth}
\centering
\includegraphics[width=0.95\textwidth]{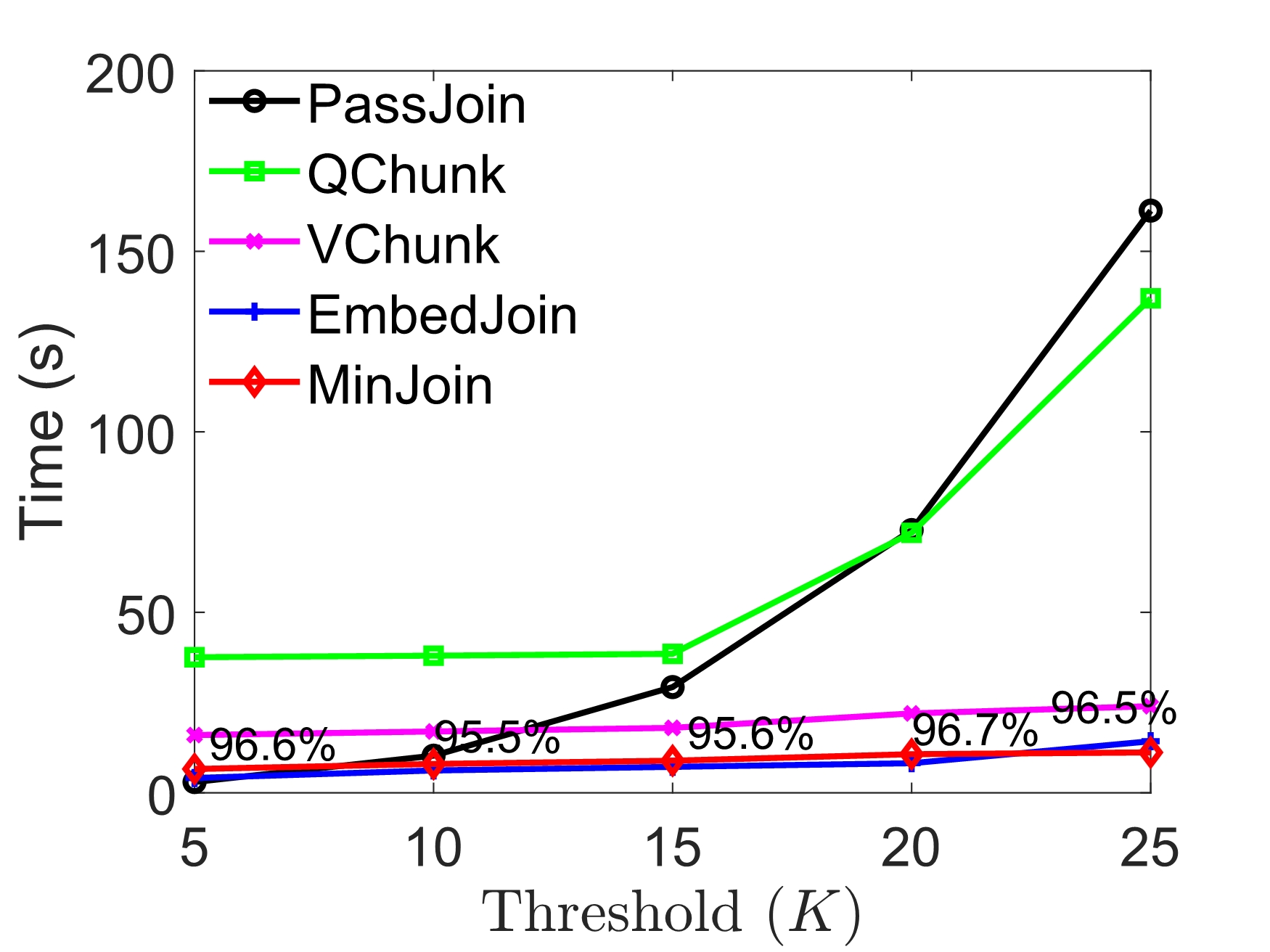}
\centerline{\uniref}
\end{minipage}
\begin{minipage}[d]{0.32\linewidth}
\centering
\includegraphics[width=0.95\textwidth]{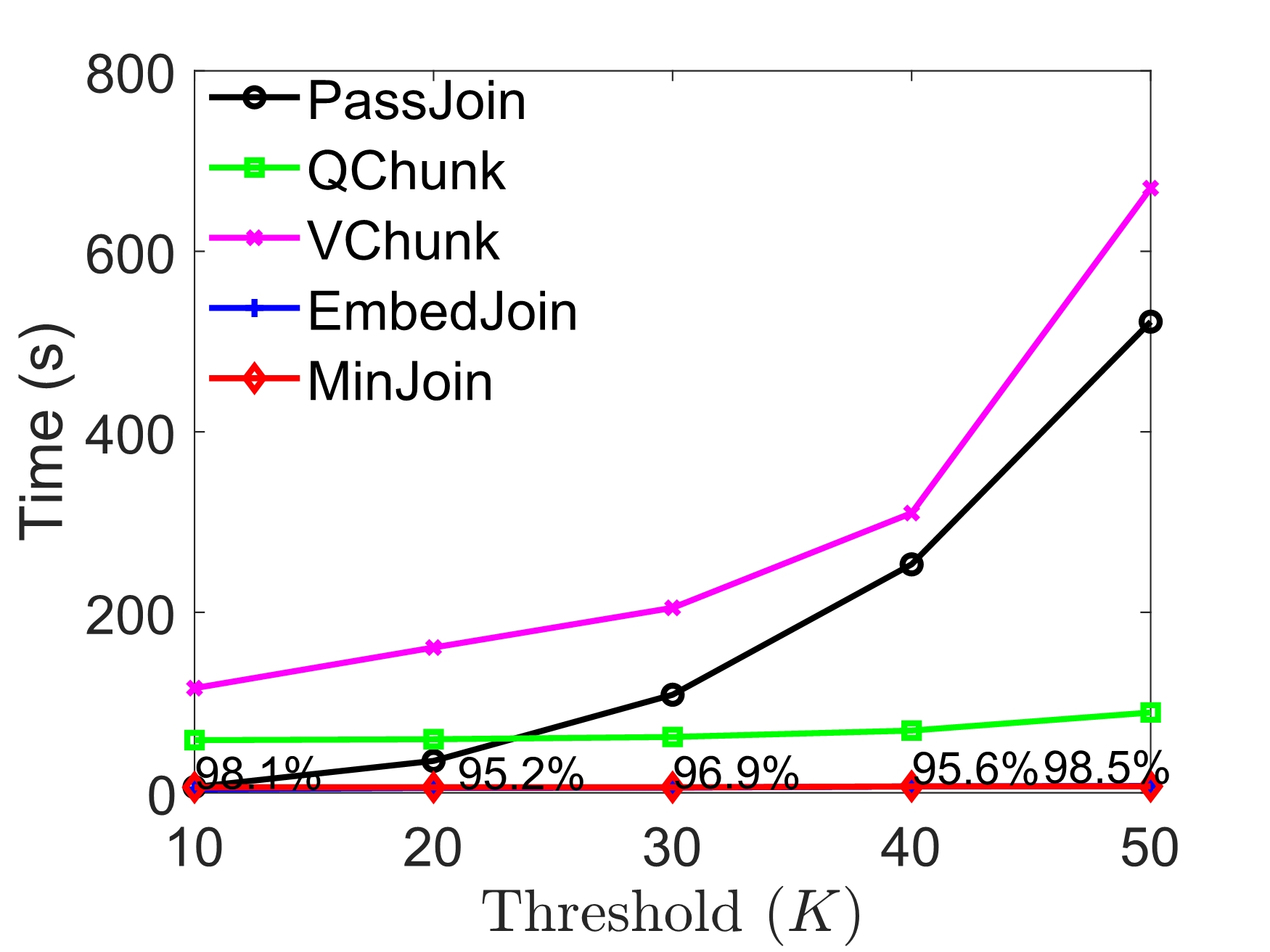}
\centerline{\trec}
\end{minipage}
\begin{minipage}[d]{0.32\linewidth}
\centering
\includegraphics[width=0.95\textwidth]{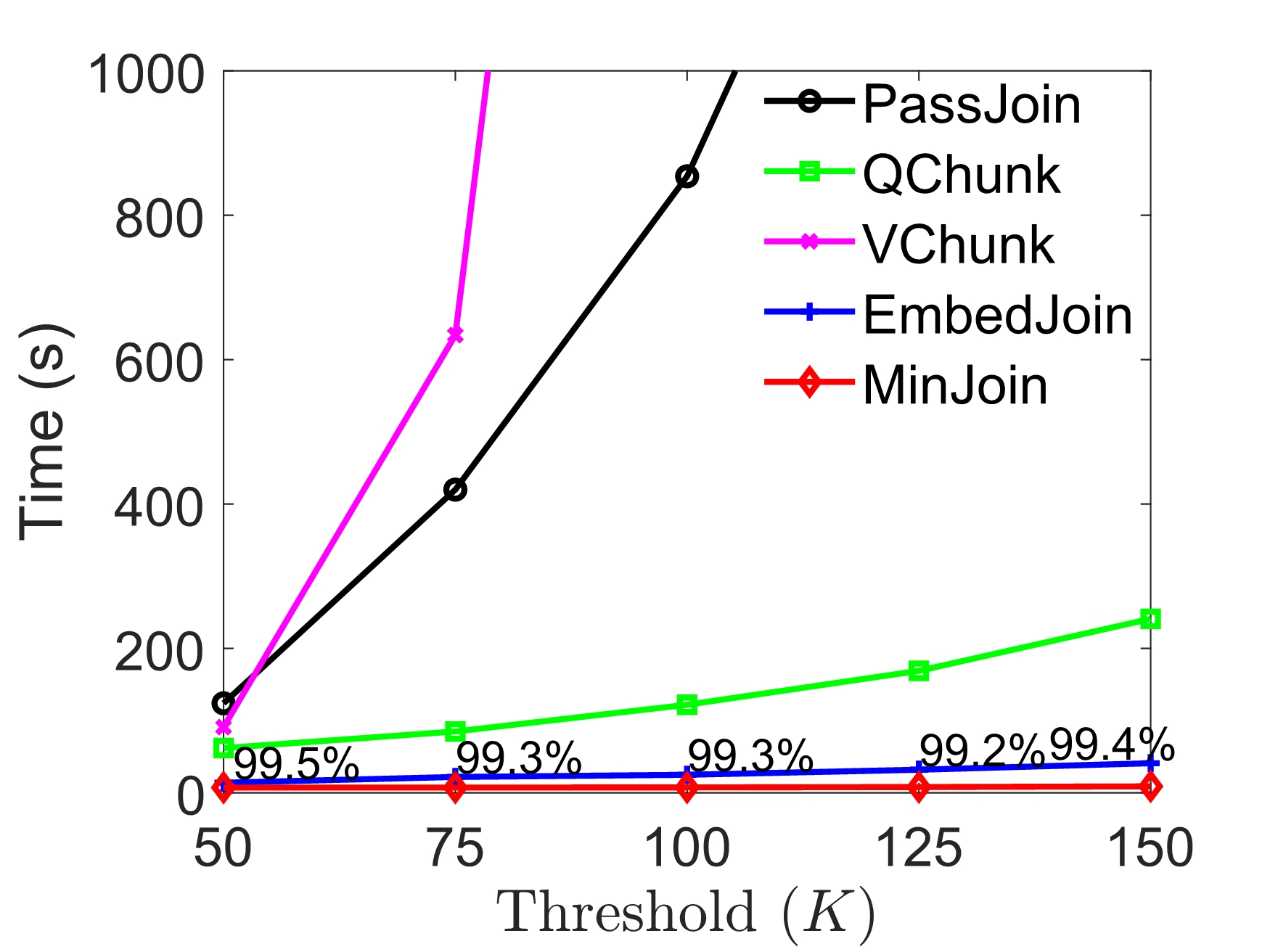}
\centerline{\genoa}
\end{minipage}
\caption{A comparison on running time, varying $K$. The percentages on plots stand for accuracy of \ebdjoin. }
\label{fig:ktime}
\end{figure*}

\begin{figure*}[t]
\centering
	\begin{minipage}[d]{0.32\linewidth}
	\centering
	\includegraphics[width=0.95\textwidth]{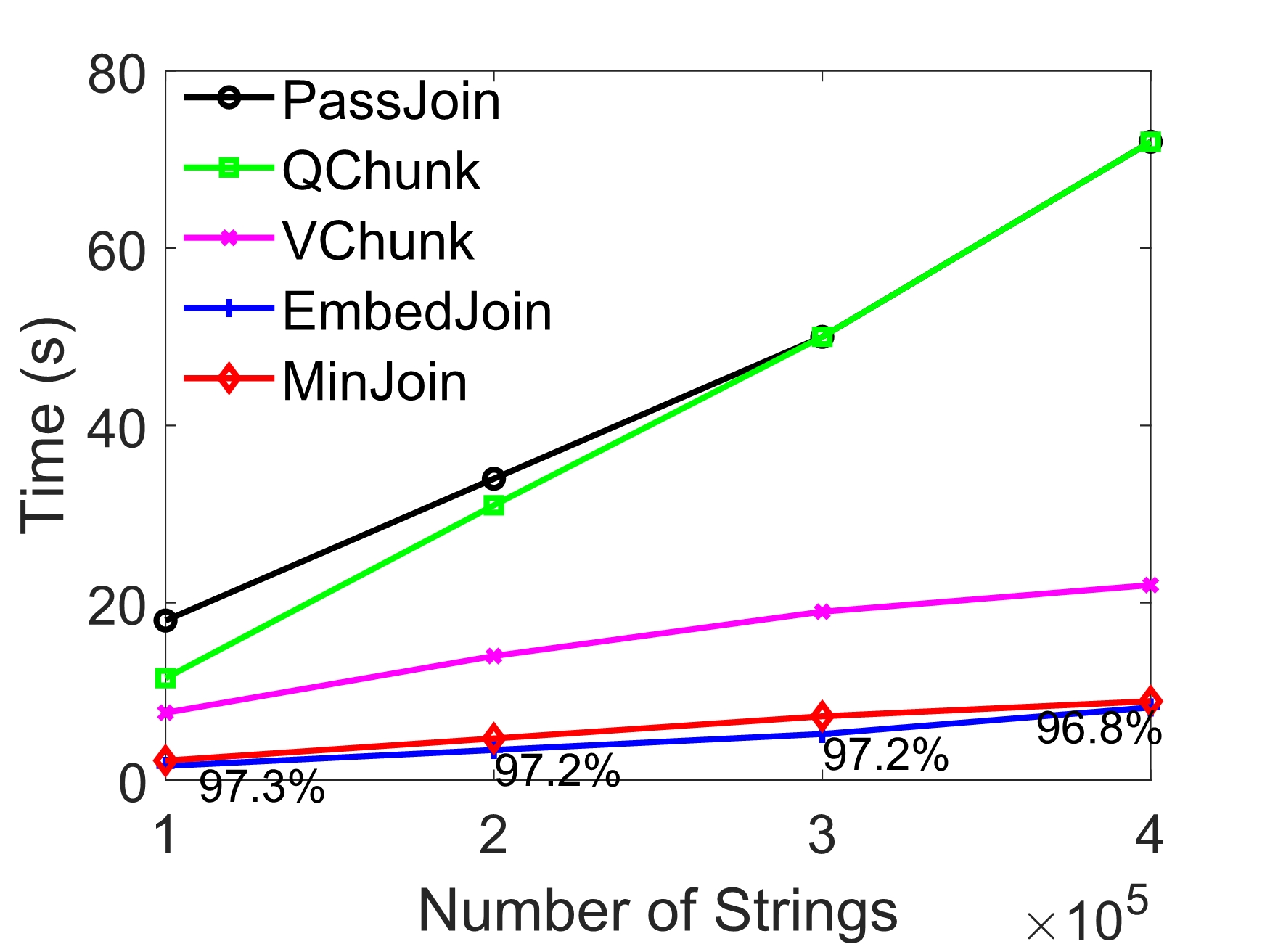}
	\centerline{\uniref\ ($K=20$)}
	\end{minipage}
	\begin{minipage}[d]{0.32\linewidth}
	\centering
	\includegraphics[width=0.95\textwidth]{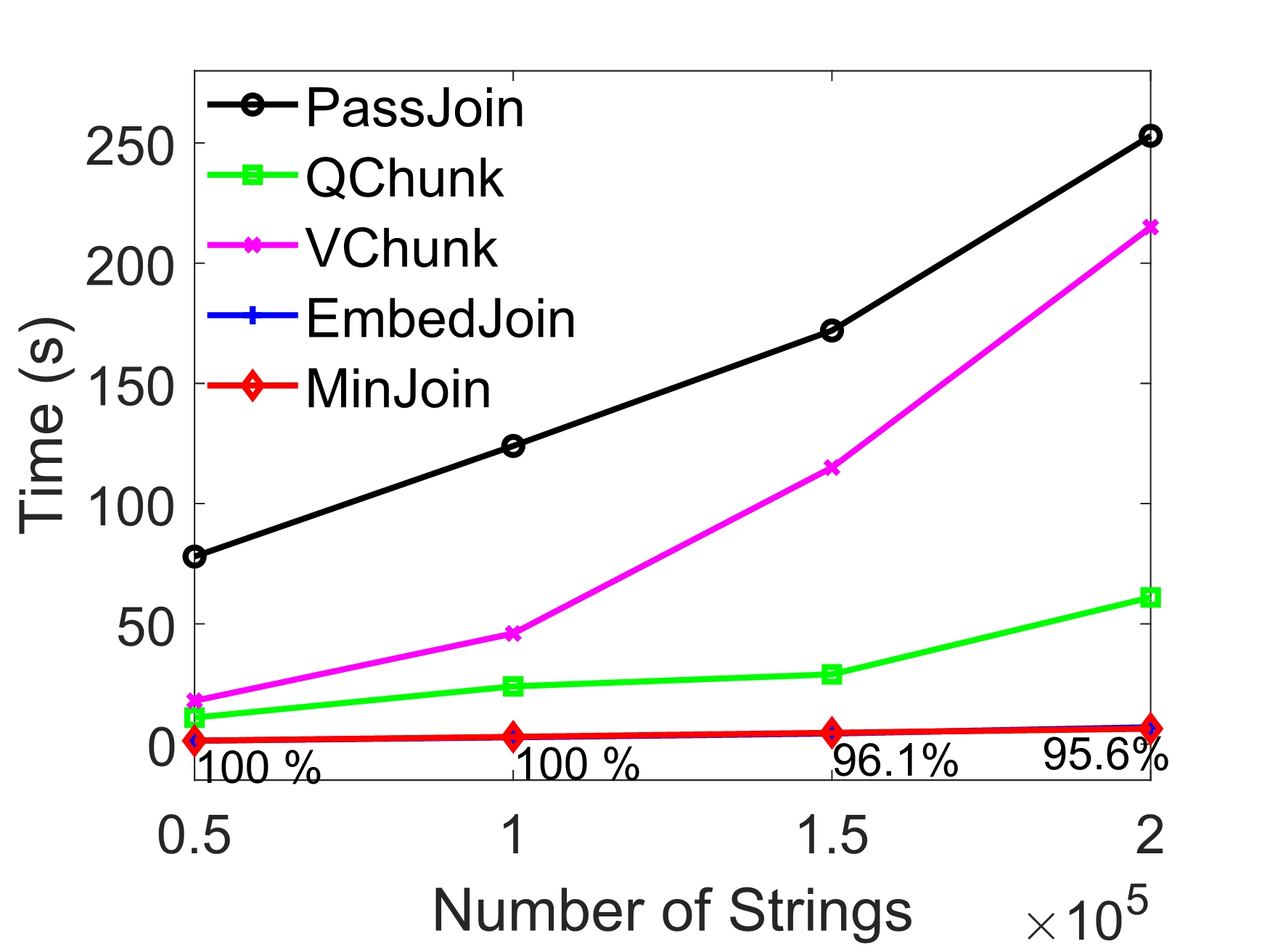}
	\centerline{\trec\ ($K=40$)}
	\end{minipage}
	\begin{minipage}[d]{0.32\linewidth}
	\centering
	\includegraphics[width=0.95\textwidth]{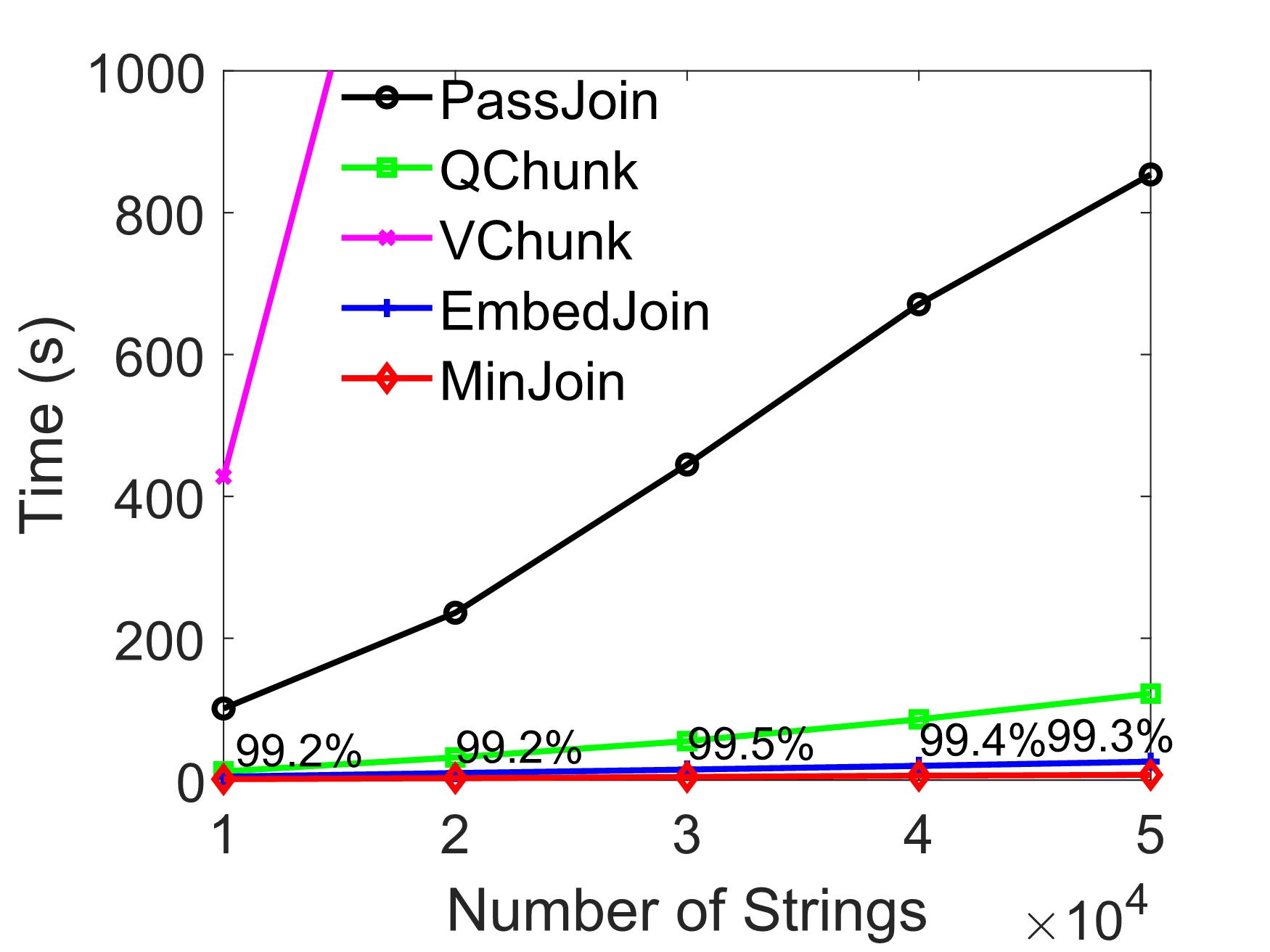}
	\centerline{\genoa\ ($K=100$)}
	\end{minipage}
	\caption{A comparison on running time, varying $n$. The percentages on plots stand for accuracy of \ebdjoin.}
	\label{fig:ntime}
	\end{figure*}

\begin{figure*}[h]
\centering
	\begin{minipage}[d]{0.32\linewidth}
	\centering
	\includegraphics[width=0.95\textwidth]{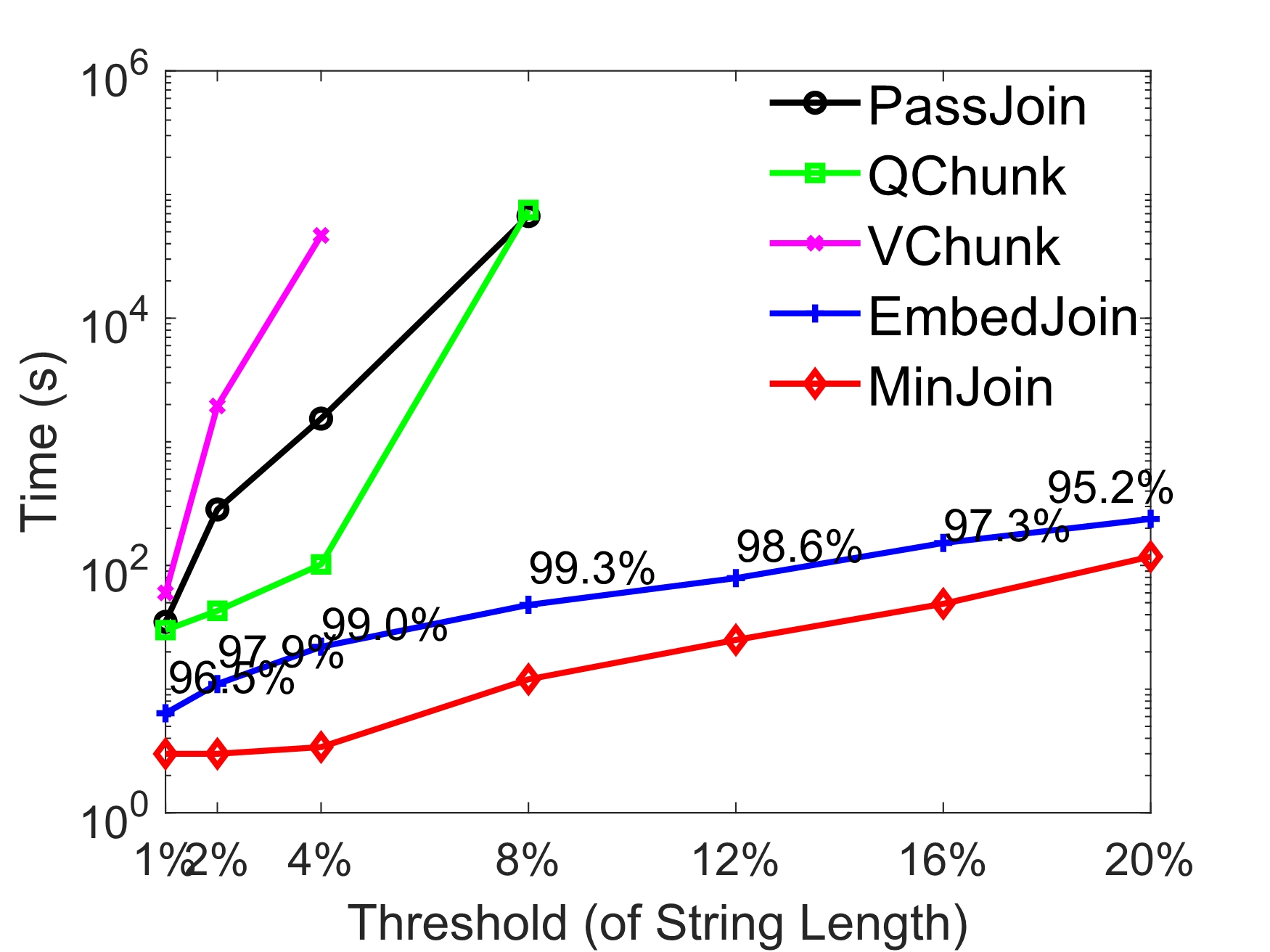}
	\centerline{\genob}
	\end{minipage}
	\begin{minipage}[d]{0.32\linewidth}
	\centering
	\includegraphics[width=0.95\textwidth]{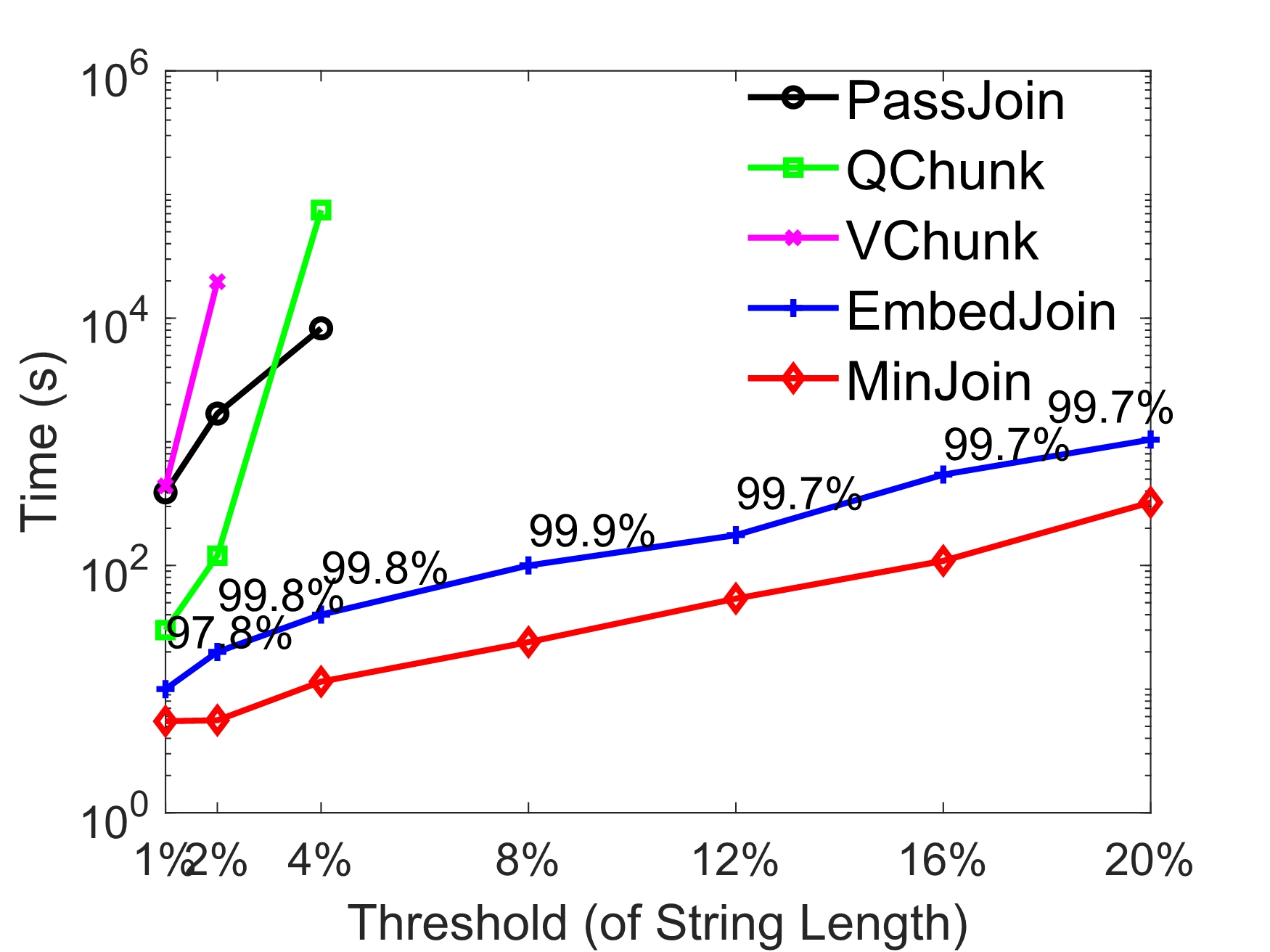}
	\centerline{\genoc}
	\end{minipage}
	\begin{minipage}[d]{0.32\linewidth}
	\centering
	\includegraphics[width=0.95\textwidth]{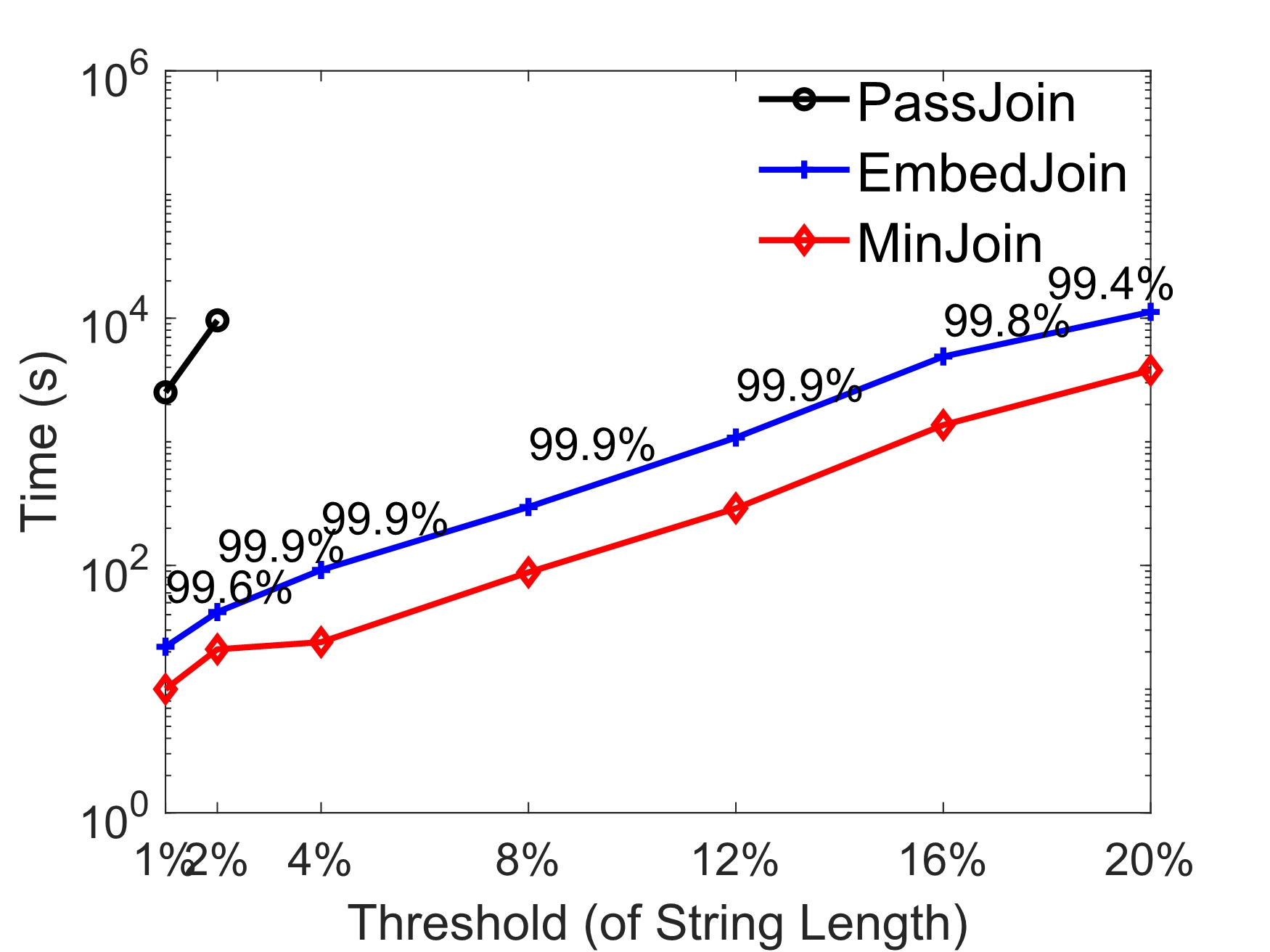}
	\centerline{\genod}
	\end{minipage}
	\caption{Scalability of different algorithms, varying $N$. The percentages on plots are accuracies of \ebdjoin.}
	\label{fig:scalek}
	\end{figure*}

	\begin{figure*}[h]
\centering
	\begin{minipage}[d]{0.32\linewidth}
	\centering
	\includegraphics[width=0.95\textwidth]{dnasmall}
	\centerline{\genob}
	\end{minipage}
	\begin{minipage}[d]{0.32\linewidth}
	\centering
	\includegraphics[width=0.95\textwidth]{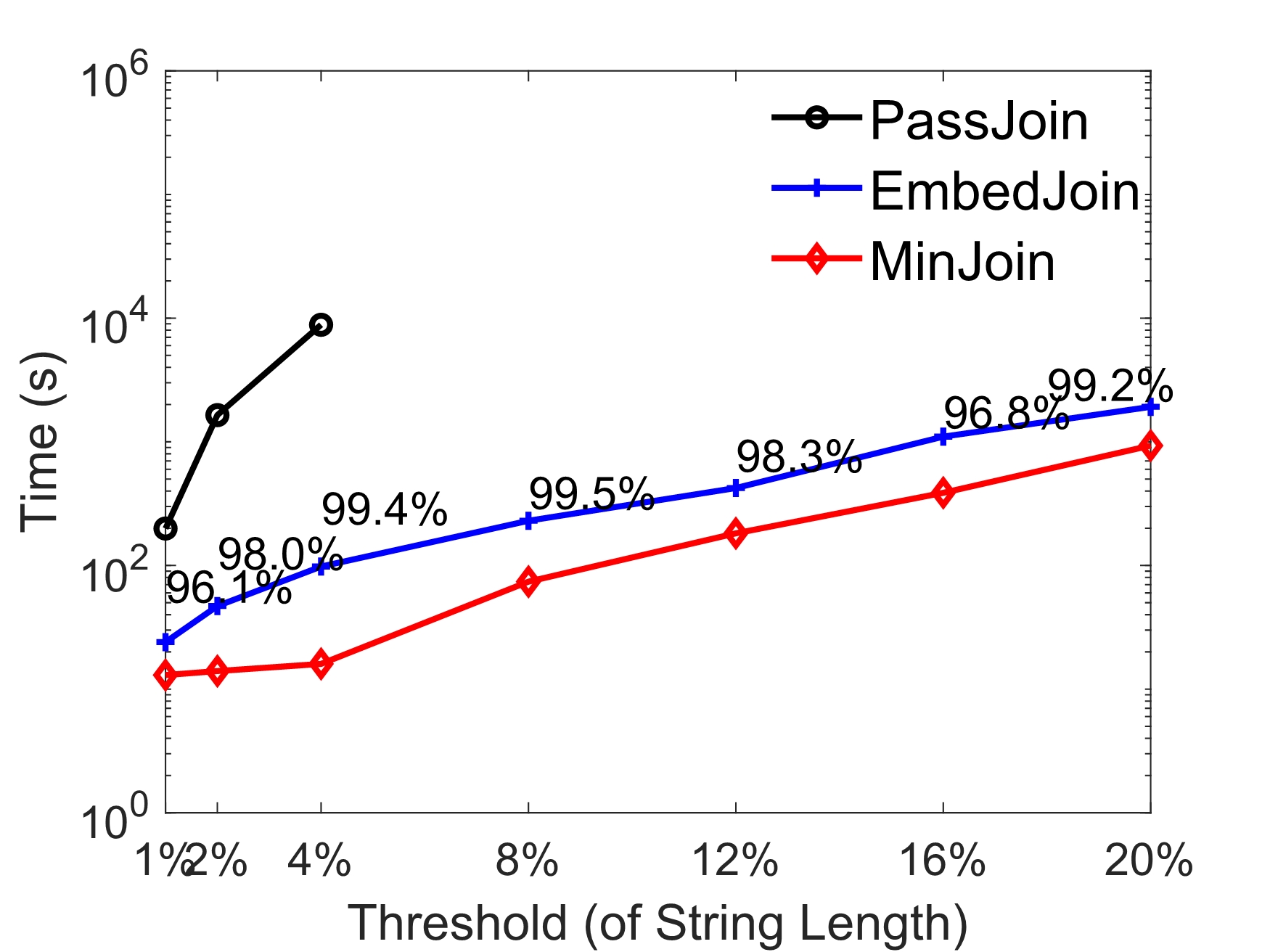}
	\centerline{\genoe}
	\end{minipage}
	\begin{minipage}[d]{0.32\linewidth}
	\centering
	\includegraphics[width=0.95\textwidth]{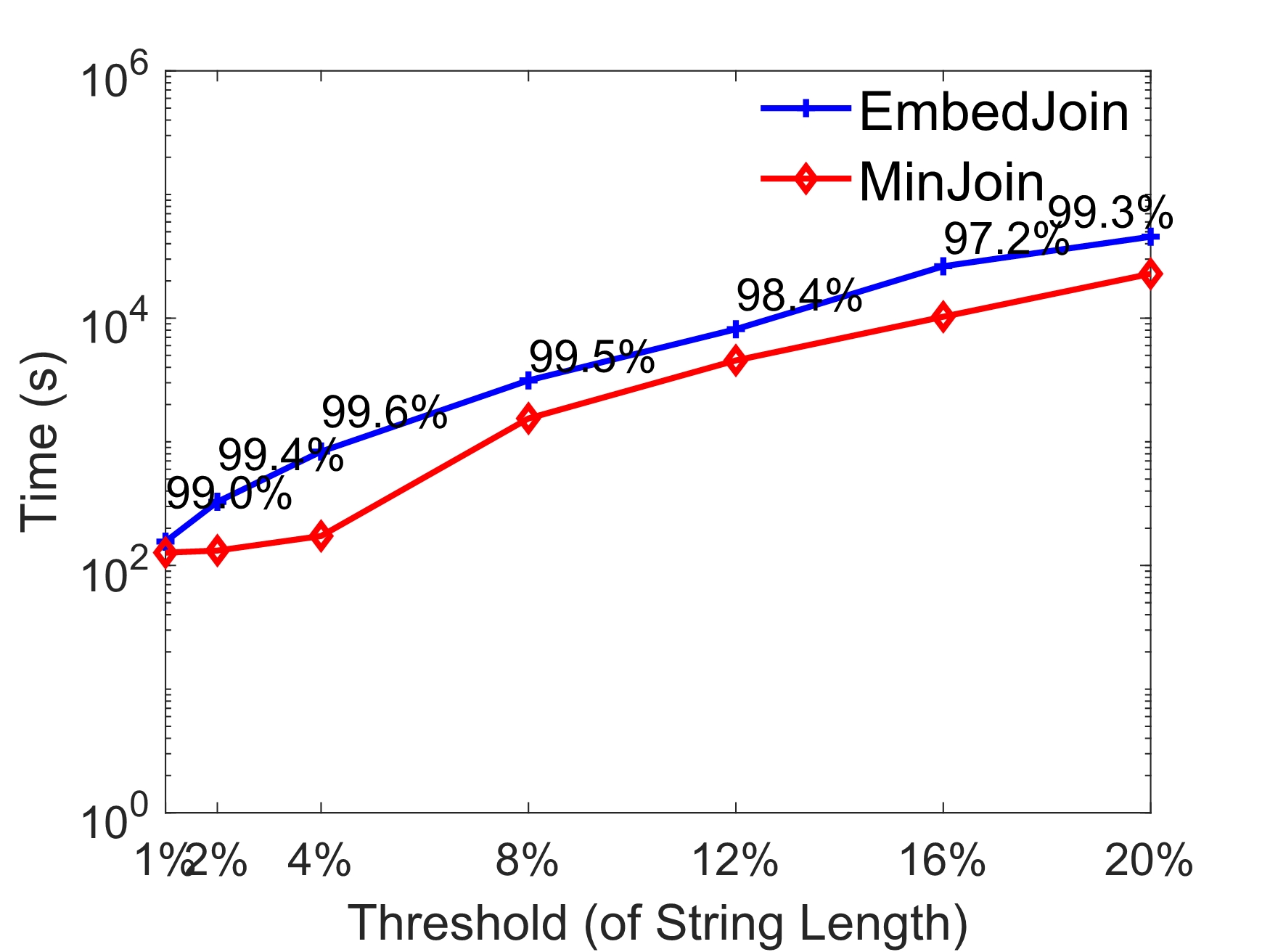}
	\centerline{\genof}
	\end{minipage}
	\caption{Scalability of different algorithms, varying $n$. The percentages on plots are accuracies of \ebdjoin.}
	\label{fig:scalen}	
	\end{figure*}

\paragraph{Effects of the Edit Threshold $K$}
Figure~\ref{fig:ktime} presents the running time of different algorithms on \uniref, \trec\ and \genoa\ when varying the edit threshold $K$.  Compared with \ebdjoin, \mjoin\ clearly has the advantage on the accuracy (100\% versus 95-99\%).  The running time of \mjoin\ is similar to \ebdjoin\ on \uniref\ and \trec, and is better than \ebdjoin\ by a factor of 4.5 on \genoa\ ($K = 150$).  We observe that \mjoin\ has a significant advantage over all the exact algorithms on running time: \mjoin\ outperforms the best exact algorithm by a factor of 2.3 in \uniref\ ($K = 25$), 12.3 on \trec\ ($K = 50$), and 26.7 on \genoa\ ($K = 150$).  The running time of \pass\ increases quickly when $K$ becomes large; this is consistent to the theory that the query time in \pass\ for each string is proportional to $K^3$.  \vchunk\ performs relatively well on \uniref,  but much worse on \trec\ and \genoa.  This may be because the preprocessing time of \vchunk\ has a quadratic dependence on string length $N$, which is larger in \trec\ and \genoa\ than \uniref. 

\paragraph{Effects of the Input Size $n$}
Figure~\ref{fig:ntime} presents the running time of different algorithms on \uniref, \trec\ and \genoa\ when varying the number of input strings $n$.  \mjoin\ again has similar running time as \ebdjoin\ on \uniref\ and \trec, and much better on \genoa\ (plus the accuracy advantage).  The running time of \mjoin\ is better than the best exact algorithm by a factor of 2.2 on \uniref\ ($n = 400,000$), 9.5 on \trec\ ($n = 200,000$), and $16.2$ on \genoa\ ($n = 50,000$).  The trends of running time of all algorithms increase near linearly in terms of $n$, except \vchunk\ whose performance deteriorates significantly when $n$ increases on \trec\ and \genoa, which may again due to the expensive preprocessing step.

\paragraph{Scalability of the Algorithms}  Finally we test all algorithms on larger datasets.  Figure~\ref{fig:scalek} presents the results of the running time when we scale string length up to 20,000 and the edit threshold $K$ up to 20\% of the string length.
Figure~\ref{fig:scalen} presents the results when we scale the number of strings up to 320,000, and $K$ up to 20\% of the string length.  The first plot of Figure~\ref{fig:scalen} is just a repeat of that of Figure~\ref{fig:scalek}. For \mjoin\ we always set the number of targeted partition $T$ to be $K/5$, which already makes the accuracy of \mjoin\ to be 100\% on those points where there is at least one exact algorithm that can finish.

We note that some algorithms cannot produce some of the points, which may be because they cannot finish within 24 hours, or there are some implementation issues (e.g., memory overflow).  In cases when there is no exact algorithm that can finish in time, the accuracy of \ebdjoin\ is computed using the result returned by \mjoin\ as the ground truth.  

We observe that \mjoin\ generally outperforms \ebdjoin\ by $2 \sim 5$ times on the running time.  The advantage slightly decreases when the number of strings $n$ or the string length $N$ increases.  This is because when $n$ or $N$ increases, the verification time ($O(NK)$ per pair where $K$ is also proportional to $N$ in our plots) will increase faster than other parts of the algorithm.  On the other hand, the accuracy of \ebdjoin, using \mjoin\ as the baseline, is about 96\%-99\%.

All the exact algorithms do not scale well on these large datasets. On the smallest dataset \genob, \pass\ and \qchunk\ can run up to the 8\% edit threshold, while \vchunk\ can only go up to the 4\% threshold.  Their running times deteriorate significantly when $K$ increases. Only \pass\ can produce some points on \genod\ and \genoe.  On \genof\ none of the exact algorithms can finish within 24 hours.

\begin{figure*}
\centering
\begin{minipage}[d]{0.32\linewidth}
\centering
\includegraphics[width=0.95\textwidth]{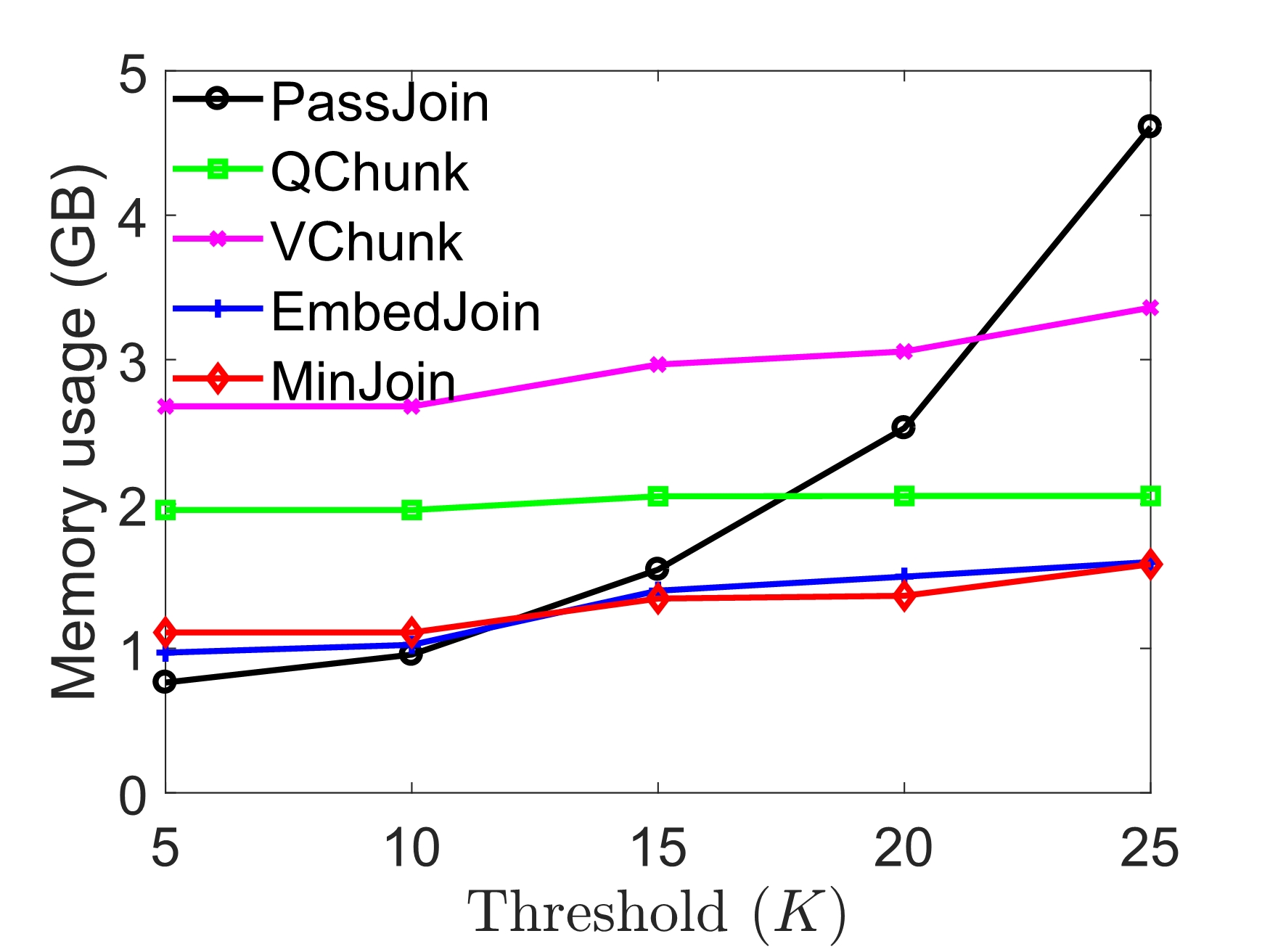}
\centerline{\uniref}
\end{minipage}
\begin{minipage}[d]{0.32\linewidth}
\centering
\includegraphics[width=0.95\textwidth]{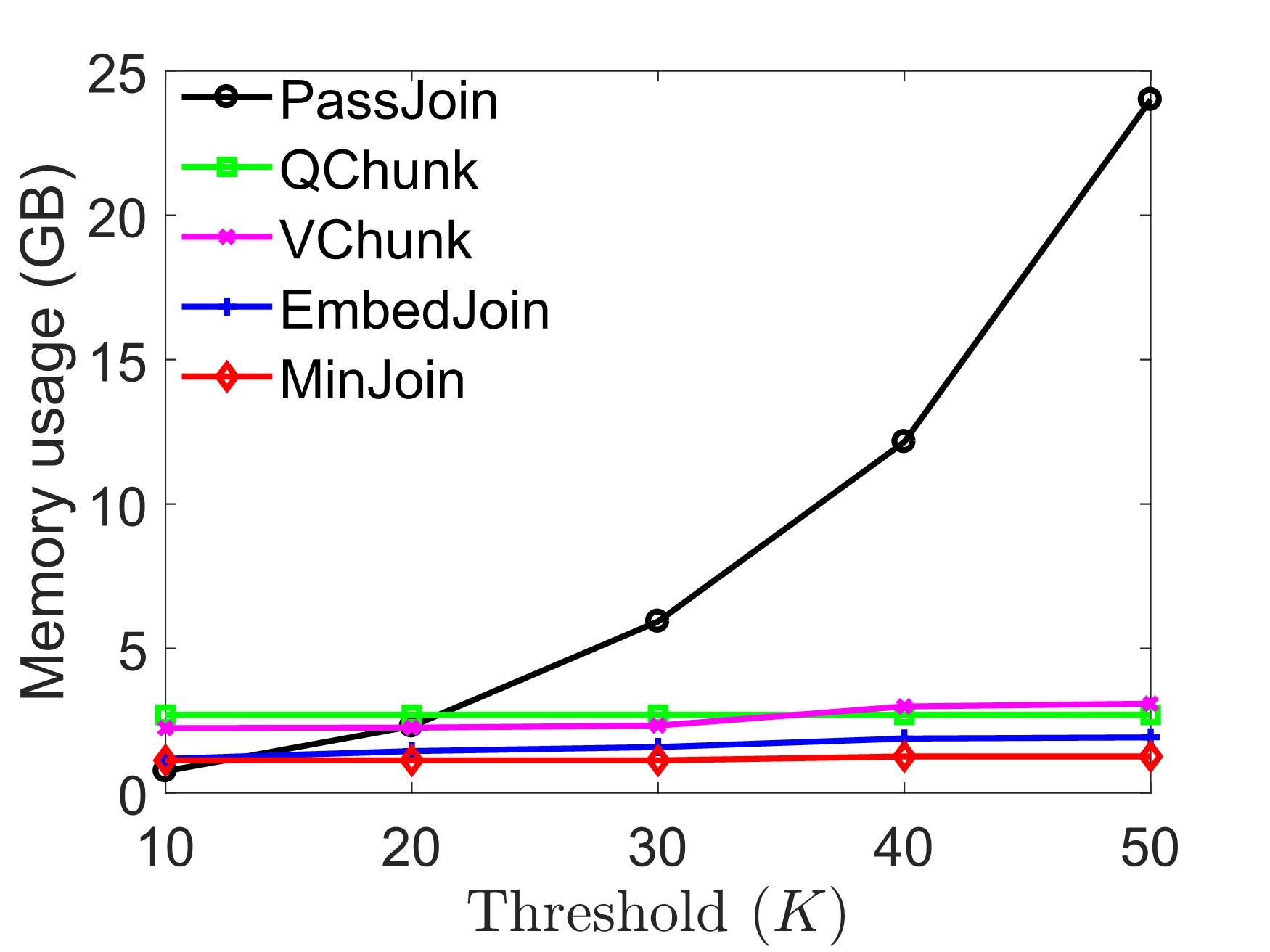}
\centerline{\trec}
\end{minipage}
\begin{minipage}[d]{0.32\linewidth}
\centering
\includegraphics[width=0.95\textwidth]{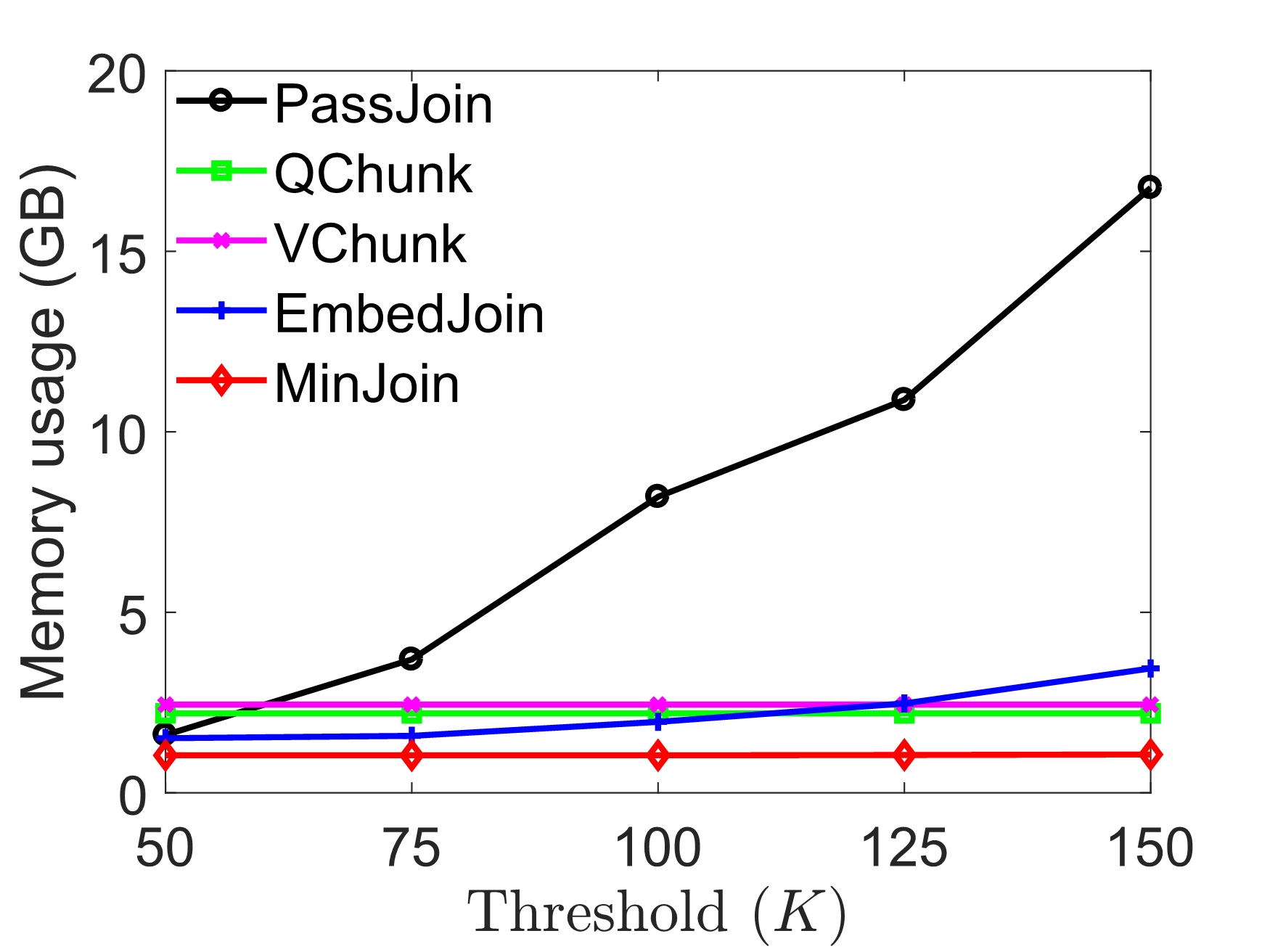}
\centerline{\genoa}
\end{minipage}
\caption{A comparison on memory usage, varying $K$. }
\label{fig:kmem}
\end{figure*}

\begin{figure*}
\centering
\begin{minipage}[d]{0.32\linewidth}
\centering
\includegraphics[width=0.95\textwidth]{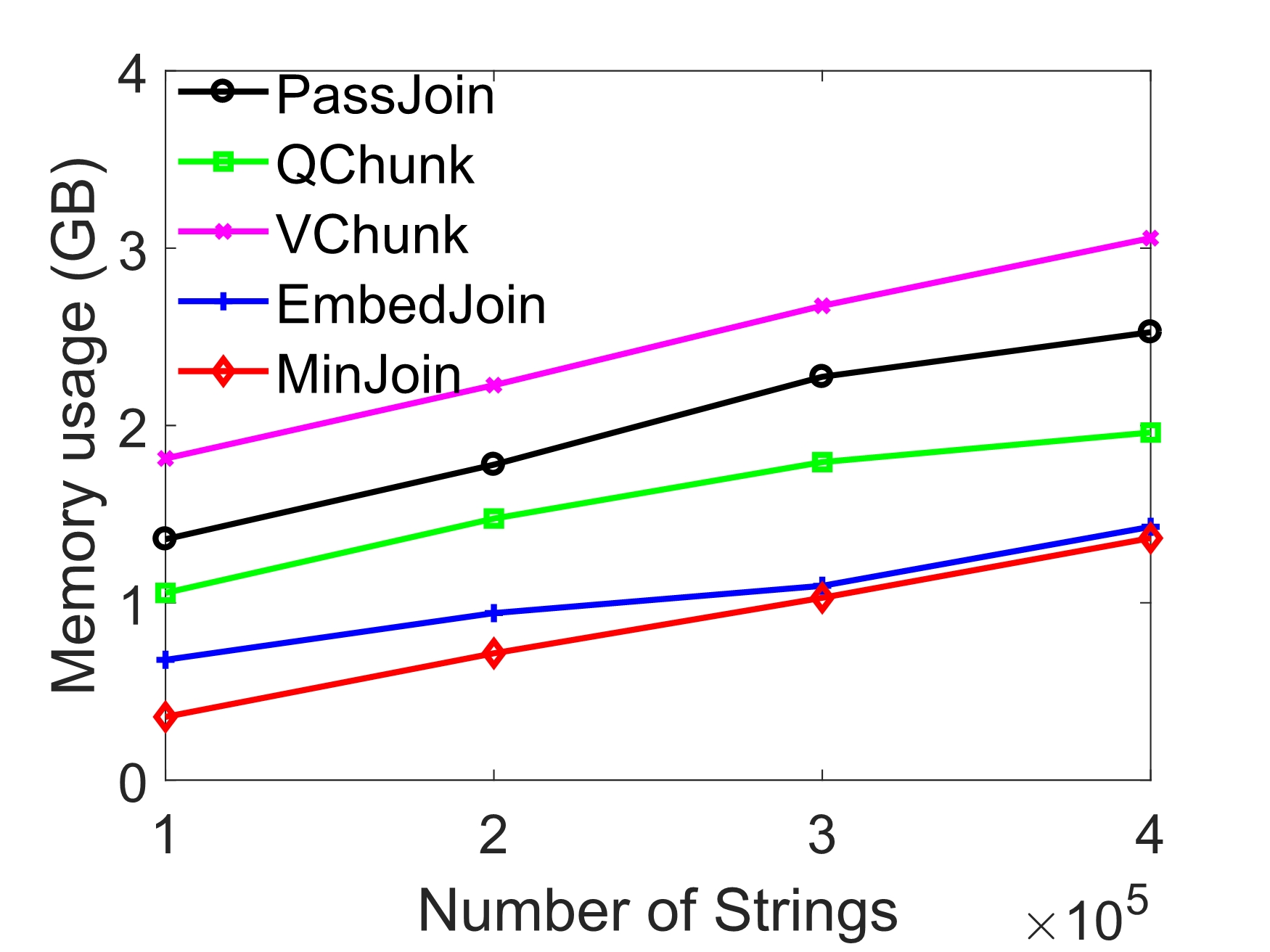}
\centerline{\uniref\ ($K=20$)}
\end{minipage}
\begin{minipage}[d]{0.32\linewidth}
\centering
\includegraphics[width=0.95\textwidth]{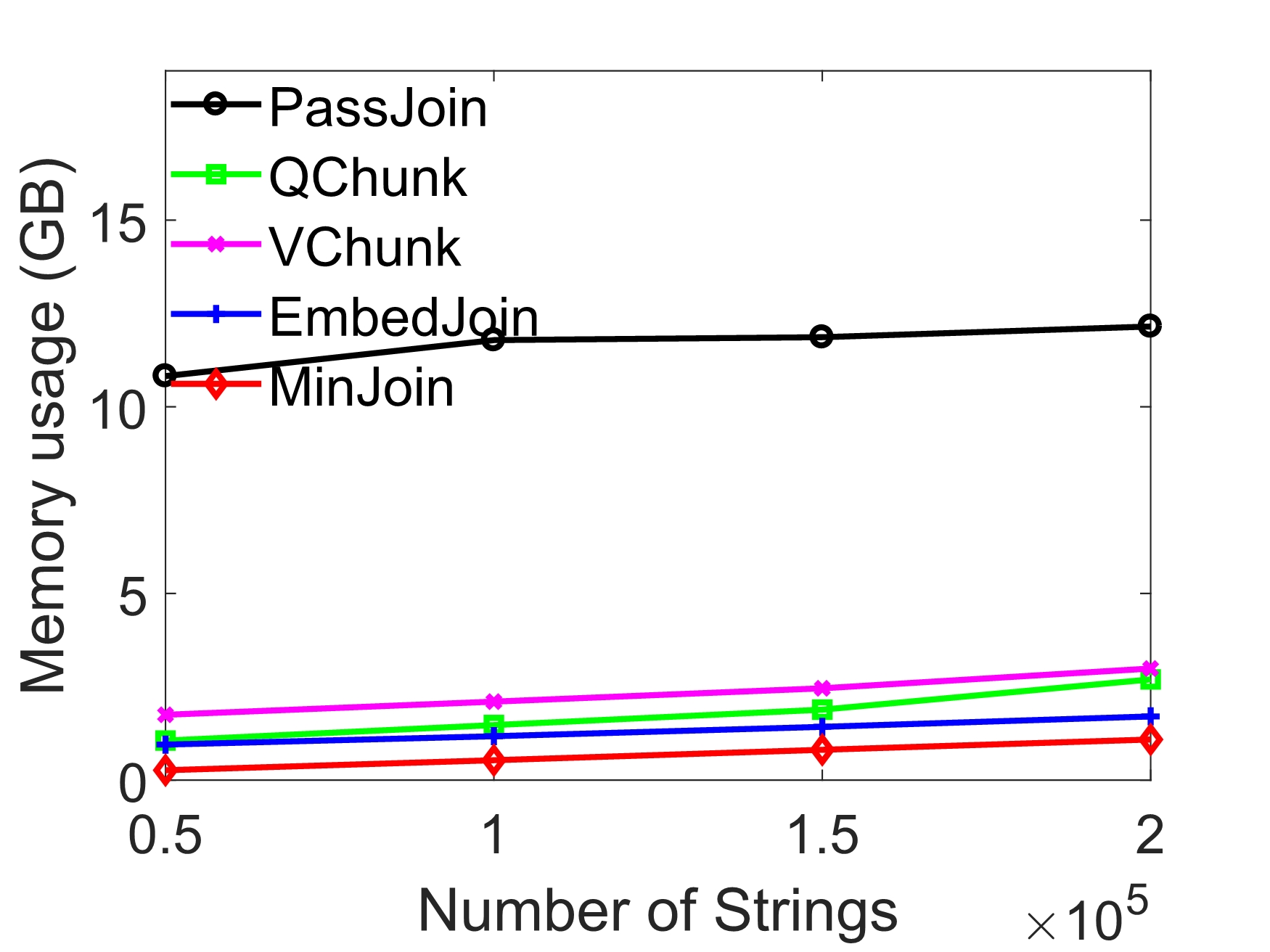}
\centerline{\trec\ ($K=40$)}
\end{minipage}
\begin{minipage}[d]{0.32\linewidth}
\centering
\includegraphics[width=0.95\textwidth]{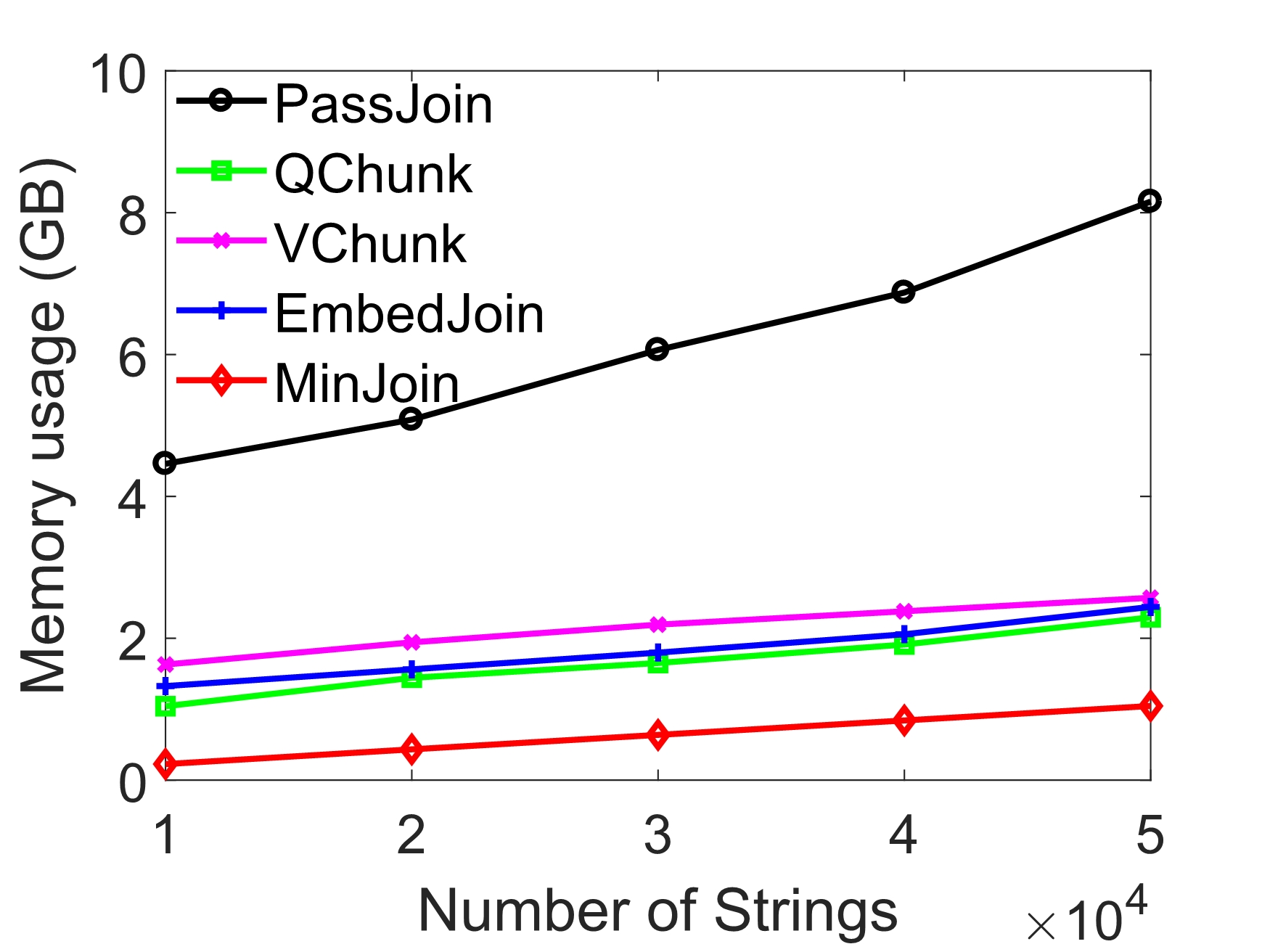}
\centerline{\genoa\ ($K=100$)}
\end{minipage}
\caption{A comparison on memory usage, varying $n$.}
\label{fig:nmem}
\end{figure*}

\paragraph{Memory Usage}  We have also compared the memory usage of all tested algorithms. Figure~\ref{fig:kmem} and Figure~\ref{fig:nmem} present the memory usage of different algorithms on \uniref, \trec\ and \genoa\ when varying edit threshold $K$ and the number of input strings $n$. While the difference on the memory usage is not as large as running time, \mjoin\ still performs the best among all algorithms. The performance of \pass\ is clearly worse than others on \trec\ and \genoa.

\section{Conclusion}
\label{sec:conclusion}

In this paper we have presented \mjoin, an algorithm for edit similarity joins based on string partition using local hash minima.  \mjoin\ has rigorous mathematical properties, and significantly outperforms previous methods on long strings with large edit thresholds.   We feel that local hash minima based string partition is a natural and elegant way for solving the edit similarity join problem: it can be applied to each string independently by a linear scan, without any synchronization between strings or global statistics of the datasets.  It also works very well with a simple hash join data structure for computing the candidate string pairs.  Moreover, even \mjoin\ is a randomized algorithm, it can easily achieve perfect accuracy on all of the datasets that we have tested.  We believe \mjoin\ is the right choice for edit similarity joins in many applications.

\bibliographystyle{acm}
\bibliography{paper}

\end{document}